\documentclass[11pt]{article}
\usepackage{mathrsfs,amssymb,amsmath,amsthm,amsfonts}
\usepackage{graphicx}
\usepackage{subfig}
\usepackage{float}

\usepackage{color}
\usepackage[colorlinks, linkcolor=red, anchorcolor=green, citecolor=blue]{hyperref}
\allowdisplaybreaks

\newtheorem{theorem}{Theorem}[section]
\newtheorem{lemma}[theorem]{Lemma}

\newtheorem{proposition}[theorem]{Proposition}
\newtheorem{remark}[theorem]{Remark}
\newtheorem{definition}[theorem]{Definition}

\numberwithin{equation}{section}
\begin{document}

\title{\Large\bf Signal Recovery under Mutual Incoherence Property and Oracle Inequalities}

\author{Peng Li\,$^\dag$ and Wengu Chen
\thanks{Corresponding author}
\thanks{P. Li is with  Graduate School, China Academy of Engineering Physics, Beijing 100088, China (E-mail: lipengmath@126.com)}
\thanks{W. Chen is with Institute of Applied Physics and Computational Mathematics, Beijing 100088, China (E-mail: chenwg@iapcm.ac.cn)}
}
\date{ }

\maketitle
\textbf{Abstract}~~This paper considers signal recovery through an unconstrained minimization in the framework of mutual incoherence property. A sufficient condition is provided to guarantee the stable recovery in the noisy case. And we give a lower bound for the $\ell_2$ norm of difference of reconstructed signals and the original signal, in the sense of expectation and probability. Furthermore, oracle inequalities of both sparse signals and non-sparse signals are derived under the mutual incoherence condition in the case of Gaussian noises. Finally, we investigate the relationship between mutual incoherence property and robust null space property and find that robust null space property can be deduced from the mutual incoherence property.



%%%%%%%%%%%%%%%%%%%%%%%%%%%%%%%%%%%%%%%%%%%%%%%%%%%%%%%%%%%%%%%%%%%%%%%%%%%%%%%%%%%%%%%%%%%%%%%%

%%%%%%%%%%%%%%%%%%%%%%%%%%%%%%%%%%%%%% Section 1 %%%%%%%%%%%%%%%%%%%%%%%%%%%%%%%%%%%%%%%%%%%%%%%

%%%%%%%%%%%%%%%%%%%%%%%%%%%%%%%%%%%%%%%%%%%%%%%%%%%%%%%%%%%%%%%%%%%%%%%%%%%%%%%%%%%%%%%%%%%%%%%%
\section{Introduction}\label{s1}
\hskip\parindent

The problem of sparse signal recovery naturally arises in genetics, communications and image processing. Prominent examples include DNA microarrays \cite{ES2005, PVMH2008}, wireless communications \cite{HS2009, THER2010}, magnetic resonance imaging \cite{LDP2007, VAHBPL2010},  and more. In such contexts, we often require to recover an unknown signal $x\in\mathbb{R}^n$ from an underdetermined system of linear equations
\begin{align}\label{systemequationsnoise}
b=Ax+z,
\end{align}
where $b\in\mathbb{R}^m$ are available measurements, the matrix $A\in\mathbb{R}^{m\times n}~(m<n)$ models the linear measurement process and
$z\in\mathbb{R}^m$ is a vector of measurement errors.

For the reconstruction of $x$, the most intuitive approach is to find the sparsest signal in the feasible set of possible solutions, i.e.,
$$
\min_{x\in\mathbb{R}^n}\|x\|_0~~\text{subject~ to}~~b-Ax\in\mathcal{B},
$$
where $\|x\|_0$ denotes the $\ell_0$ norm of $x$, i.e., the number of nonzero coordinates, and $\mathcal{B}$ is a bounded set determined by the error structure. However, such method is NP-hard and thus computationally infeasible in high dimensional sets.  Cand\`{e}s and Tao \cite{CT2005} proposed a convex relaxation of this method-the constrained $\ell_1$ minimization method. It estimates the signal $x$ by
\begin{align}\label{l1minimization}
\hat{x}=\arg\min_{x\in\mathbb{R}^n}\{\|x\|_1: b-Ax\in\mathcal{B}\}.
\end{align}
We often consider two types of bounded noises. One is $l_2$ bounded noises \cite{DET2005}, i.e.,
\begin{align}\label{QCBPmodel}
\min_{x\in\mathbb{R}^n}~\|x\|_1 ~~\text{subject~ to}~~\|b-Ax\|_{2}\leq\eta
\end{align}
for some constant $\eta$, which is called quadratically constrained basis pursuit (QCBP). And the other is motivated by \textit{Dantzig~selector} procedure \cite{CT2007}, i.e.,
\begin{align}\label{DantzigselectorMedel}
\min_{x\in\mathbb{R}^n}~\|x\|_1 ~~\text{subject~ to}~~\|A^*(b-Ax)\|_{\infty}\leq\eta.
\end{align}
In particular, when $\eta=0$, it is the noiseless case as follows
\begin{align}\label{NoiselessMedel}
\min_{x\in\mathbb{R}^n}~\|x\|_1~~\text{subject~ to}~~Ax=b,
\end{align}
which is called basis pursuit (BP) \cite{CDS1998}.

For the $\ell_1$-minimization problem (\ref{l1minimization}), there are many works under the restricted isometry property \cite{CT2005,D2006,CT2006,CT2007,CDD2009, CZ2013,CZ2014,ZL2017} and under the null space property \cite{DE2003, CDD2009, S2011, FR2013, F2014}. We shall not conduct a review here, as this paper is concerned with a different approach.

Here we consider recovering a signal under the framework of mutual incoherence property (MIP),  a regularity and widely used condition. It is introduced in \cite{DH2001} and defined as follows.

\begin{definition}\label{MIPDefinition}
Let $A\in\mathbb{R}^{m\times n}$ be a matrix with $\ell_2$-normalized columns $A_1,\ldots,A_n$, i.e., $\|A_i\|_2=1$ for all $i=1,\ldots,n$. The coherence $\mu=\mu(A)$ of matrix $A$ is defined as
\begin{align}
\mu=\max_{1\leq i\neq j\leq n}|\langle A_i, A_j \rangle|.
\end{align}
When coherence $\mu$ is small, we say that $A$ satisfies mutual incoherence property.
\end{definition}

It was first shown by Donoho and Huo \cite{DH2001}, in the noiseless case for the setting where $A$ is a concatenation of two square
orthogonal matrices, that
\begin{align}\label{SharpMIP}
\mu<\frac{1}{2s-1}
\end{align}
ensures the exact recovery of $x$ when $x$ is $s$-sparse. And in the noisy case, Donoho, Elad and Temlyakov \cite{DET2005} showed that when $\mathcal{B}=\mathcal{B}^{\ell_2}$, sparse signals can be recovered approximately with the error at worst proportional to
the input noise level, under the condition $\mu<1/(4s-1)$.
More results under the framework of MIP, readers can refer to \cite{T2004,CXZ2009,T2009}.

Instead of solving (\ref{QCBPmodel}) directly, many algorithms were proposed to solve the following unconstrained problem
\begin{align}\label{LassoModel}
\min_{x\in\mathbb{R}^n}~\lambda\|x\|_1+\frac{1}{2}\|Ax-b\|_2^2,
\end{align}
which is called Lasso and introduced in \cite{T1996}. There are many works under the condition of restricted isometry property for this Lasso model, readers can refer to \cite{DET2005,EMR2007,BRT2009,D2013,TEBN2014,LL2014,XL2016,ZYY2016}. We should point out that in \cite{SHB2015}, Shen, Han and Breverman showed that if $\delta_{2s}<1/5$, then $x$ can be stably recovered via analysis based approaches.
But as far as we know, there lacks MIP based theoretical study about Lasso. Therefore, we purse the MIP analysis of Lasso in this paper.

The paper is organized as follows. In Section \ref{s2}, we  establish that if $\mu<1/(4s)$, then the signal $x$ can be stably
recovered through unconstrained minimization (\ref{LassoModel}). And we also obtain the lower bound for the $\ell_2$ norm of difference of reconstructed signals and the original signal in the sense of expectation and probability in Section  \ref{s2}. In Section \ref{s3}, we show oracle inequalities for both sparse signals and non-sparse signals. And in Section \ref{s4}, we study the relationship between MIP and robust null space property and obtain that
$\ell_2$ robust null space property of order $s$ can be deduced from MIP if $\mu<1/\big(\sqrt{2}(2s-1)\big)$.
Finally, we summarize our conclusions in Section \ref{s5}.

Throughout the article, we use the following basic notation.
For $x\in\mathbb{R}^n$, we denote $x_{\max(s)}$ as the vector $x$ with all but the largest $s$ entries in absolute value set to zero, and $x_{-\max(s)}=x-x_{\max(s)}$.  Let $S$ denotes a index set, $\chi_{S}$ denotes the characteristic function on $S$, i.e., if $j\in S$, then $\chi_{S}(j)=1$; otherwise $\chi_{S}(j)=0$. And let $x_S$ be the vector equal to $x$ on $S$ and to zero on $S^c$.

%%%%%%%%%%%%%%%%%%%%%%%%%%%%%%%%%%%%%%%%%%%%%%%%%%%%%%%%%%%%%%%%%%%%
%%%%%%%%%%%%%%%%%%%%%%   section 2 %%%%%%%%%%%%%%%%%%%%%%%%%%%%%%%%%
%%%%%%%%%%%%%%%%%%%%%%%%%%%%%%%%%%%%%%%%%%%%%%%%%%%%%%%%%%%%%%%%%%%
\section{Stable Recovery \label{s2}}
\hskip\parindent

Now, we consider the stable recovery of sparse and non-sparse signals $x\in\mathbb{R}^{n}$ through Lasso model (\ref{LassoModel}). It will be shown that the condition $\mu<1/(4s)$ is sufficient for the stably recovery in the noisy case in Subsection \ref{s2.1}.
And although Cai, Wang and Zhang \cite{CWX2010} pointed out that QCBP model (\ref{QCBPmodel}) and Dantzig selector model (\ref{DantzigselectorMedel}) can also recover $x$ with accuracy if $x$ has good $s$-term approximation under the condition $\mu<1/(2s-1)$, they did not give out exact expression. In order to apply this result to oracle inequality in general case (Subsection \ref{s3.2}), we will give out the result without proof in Subsection \ref{s2.2}.
Finally, considering the Gaussian noise, we give out the lower bound of the $\ell_2$ norm of difference of reconstructed signals and the original signal in the sense of expectation and probability in Subsection \ref{s2.3}.

\subsection{Stable Recovery of Lasso \label{s2.1}}
\hskip\parindent

In this subsection, we consider the stable recovery of signals through Lasso model (\ref{LassoModel}).

\begin{theorem}\label{LassoTheorem}
Assume the measurement matrix $A$ satisfies the MIP with
$$
\mu<\frac{1}{4s}
$$
and $\|A^*z\|_{\infty}\leq\lambda/2$.  Let $\hat{x}^{L}$ be the solution to the Lasso (\ref{LassoModel}), then
\begin{align*}
\|\hat{x}^{L}-x\|_2\leq\frac{15\sqrt{s}}{8\mu(1-4s\mu)}\lambda+\Big(\frac{2(1+2s)\mu}{1-4s\mu}+\frac{1}{2}\Big)\frac{2\|x_{-\max(s)}\|_1}{\sqrt{s}}.
\end{align*}
\end{theorem}

Before giving out the proof, we first recall two auxiliary lemmas. The first one is a modified cone constraint inequality (see, e.g., \cite[Page 2356]{CP2011} for matrix case and \cite[Lemma 2]{SHB2015} for vector with frame).

\begin{lemma}\label{LassoConeconstriantinequalityLemma}
If the noisy measurements $b=Ax+z$ are observed with noise level $\|A^*z\|_{\infty}\leq\lambda/2$, then the minimization solution $\hat{x}^{L}$ of (\ref{LassoModel}) satisfies
\begin{align*}
\|Ah\|_2^2+\lambda\|h_{-\max(s)}\|_1\leq 3\lambda\|h_{\max(s)}\|_1+4\lambda\|x_{-\max(s)}\|_1,
\end{align*}
where $h=\hat{x}^{L}-x$. In particular,
\begin{align*}
\|h_{-\max(s)}\|_1\leq 3\|h_{\max(s)}\|_1+4\|x_{-\max(s)}\|_1
\end{align*}
and
\begin{align*}
\|Ah\|_2^2\leq3\lambda\|h_{\max(s)}\|_1+4\lambda\|x_{-\max(s)}\|_1.
\end{align*}
\end{lemma}

The second one is a useful property of MIP, see \cite{CWX2010}.
\begin{lemma}\label{MIPLemma}
Let $A\in\mathbb{R}^{m\times n}$ be a matrix with $\ell_2$-normalized columns and $s\in\{1,\ldots,n\}$. For all $s$-sparse vectors $x\in\mathbb{R}^n$,
\begin{align*}
\big(1-(s-1)\mu\big)\|x\|_2^2\leq\|Ax\|_2^2\leq\big(1+(s-1)\mu\big)\|x\|_2^2.
\end{align*}
\end{lemma}

Now, we are in position of proving our result.
\begin{proof}[Proof of Theorem \ref{LassoTheorem}]
Suppose that $\text{supp}(x_{\max(s)})\subset S$ with $|S|=s$. Set $h=\hat{x}^{L}-x$. By Lemma \ref{LassoConeconstriantinequalityLemma}, we have cone constraint inequality
\begin{align}\label{LassoConeconstriantinequality2}
\|h_{-\max(s)}\|_1\leq 3\|h_{\max(s)}\|_1+4\|x_{-\max(s)}\|_1
\end{align}
and
\begin{align}\label{Coneconstriantinequality1}
\|Ah\|_2^2\leq 3\lambda\|h_{\max(s)}\|_1+4\lambda\|x_{-\max(s)}\|_1.
\end{align}

By
\begin{align*}
\|h\|_2=\sqrt{\|h_{\max(s)}\|_2^2+\|h_{-\max(s)}\|_2^2},
\end{align*}
we need to estimate $\|h_{\max(s)}\|_2$ and $\|h_{-\max(s)}\|_2$, respectively.

To estimate $\|h_{\max(s)}\|_2$, we consider the following identity
\begin{align}\label{e2.1}
\big|\langle Ah, Ah_{\max(s)}\rangle\big|=\big|\langle Ah_{\max(s)},Ah_{\max(s)}\rangle +\langle Ah_{-\max(s)},Ah_{\max(s)}\rangle\big|.
\end{align}
First, we give out a lower bound for (\ref{e2.1}). It follows from Lemma \ref{MIPLemma} and Definition \ref{MIPDefinition} that
\begin{align}\label{e2.2}
\big|\langle Ah, Ah_{\max(s)}\rangle\big|
&\geq\big|\langle Ah_{\max(s)},Ah_{\max(s)}\rangle\big|-\big|\langle Ah_{-\max(s)},Ah_{\max(s)}\rangle\big|\nonumber\\
&=\big|\langle Ah_{\max(s)},Ah_{\max(s)}\rangle\big|-\sum_{i\in S}\sum_{j\in S^c}|\langle A_i,A_j\rangle|| h(i)h(j)|\nonumber\\
&\geq \big(1-(s-1)\mu\big)\|h_{\max(s)}\|_2^2-\mu\|h_{\max(s)}\|_1\|h_{-\max(s)}\|_1.
\end{align}
Then by (\ref{LassoConeconstriantinequality2}), we can get a lower bound of $\big|\langle Ah, Ah_{\max(s)}\rangle\big|$ as follows
\begin{align}\label{e2.3}
\big|\langle Ah&, Ah_{\max(s)}\rangle\big|\nonumber\\
&\geq \big(1-(s-1)\mu\big)\|h_{\max(s)}\|_2^2-\mu\|h_{\max(s)}\|_1\big(3\|h_{\max(s)}\|_1+4\|x_{-\max(s)}\|_1\big)\nonumber\\
&\geq\big(1-(s-1)\mu\big)\|h_{\max(s)}\|_2^2-3\mu(\sqrt{s}\|h_{\max(s)}\|_2)^2-4\mu\|x_{-\max(s)}\|_1(\sqrt{s}\|h_{\max(s)}\|_2)\nonumber\\
&=\big(1-(4s-1)\mu\big)\|h_{\max(s)}\|_2^2-2s\mu\frac{2\|x_{-\max(s)}\|_1}{\sqrt{s}}\|h_{\max(s)}\|_2.
\end{align}

Next, we also have an upper bound for (\ref{e2.1}). Using (\ref{Coneconstriantinequality1}) and Lemma \ref{MIPLemma}, we obtain
\begin{align*}
\big|\langle Ah&, Ah_{\max(s)}\rangle\big|\leq\|Ah\|_{2}\|Ah_{\max(s)}\|_{2}\nonumber\\
&\leq\sqrt{3\lambda\|h_{\max(s)}\|_1+4\lambda\|x_{-\max(s)}\|_1}\sqrt{(1+(s-1)\mu)\|h_{\max(s)}\|_2^2}\nonumber\\
&\leq\sqrt{\frac{1}{\varepsilon}\lambda\sqrt{s}\bigg(3\varepsilon\|h_{\max(s)}\|_2+4\varepsilon\frac{\|x_{-\max(s)}\|_1}{\sqrt{s}}\bigg)}
\sqrt{1+(s-1)\mu}\|h_{\max(s)}\|_2,
\end{align*}
where $\varepsilon>0$ is to be determined. Then the elementary inequality $\sqrt{|a||b|}\leq(|a|+|b|)/2$ implies that
\begin{align*}
\big|\langle Ah, Ah_{\max(s)}\rangle\big|
&\leq\frac{3\varepsilon\|h_{\max(s)}\|_2+4\varepsilon\frac{\|x_{-\max(s)}\|_1}{\sqrt{s}}+\frac{1}{\varepsilon}\sqrt{s}\lambda}{2}
\sqrt{1+(s-1)\mu}\|h_{\max(s)}\|_2\\
&=\frac{3\varepsilon\sqrt{1+(s-1)\mu}}{2}\|h_{\max(s)}\|_2^2\\
&\hspace*{12pt}+\bigg(2\varepsilon\sqrt{1+(s-1)\mu}\frac{\|x_{-\max(s)}\|_1}{\sqrt{s}}
+\frac{\sqrt{1+(s-1)\mu}}{2\varepsilon}\sqrt{s}\lambda\bigg)
\|h_{\max(s)}\|_2.
\end{align*}
Taking $\varepsilon=2\mu/\big(3\sqrt{1+(s-1)\mu}\big)$, then we have
\begin{align}\label{e2.4}
\big|\langle Ah, Ah_{\max(s)}\rangle\big|
&\leq\mu\|h_{\max(s)}\|_2^2+\bigg(\frac{2\mu}{3}\frac{2\|x_{-\max(s)}\|_1}{\sqrt{s}}+\frac{3(1+(s-1)\mu)}{4\mu}\sqrt{s}\lambda\bigg)\|h_{\max(s)}\|_2.
\end{align}

Combining the lower bound (\ref{e2.3}) with the upper bound (\ref{e2.4}), we get
\begin{align}\label{e2.5}
(1-4s\mu)\|h_{\max(s)}\|_2^2-\Bigg(\bigg(\frac{2\mu}{3}+2s\mu\bigg)\frac{2\|x_{-\max(s)}\|_1}{\sqrt{s}}+\frac{3(1+(s-1)\mu)}{4\mu}\sqrt{s}\lambda\Bigg)
\|h_{\max(s)}\|_2\leq0.
\end{align}
Note that $\mu<1/(4s)$. Therefore
\begin{align}\label{e2.6}
\|h_{\max(s)}\|_2&\leq\bigg(\frac{2\mu}{3}+2s\mu\bigg)\frac{1}{1-4s\mu}\frac{2\|x_{-\max(s)}\|_1}{\sqrt{s}}
+\frac{3(1+(s-1)\mu)}{4\mu(1-4s\mu)}\sqrt{s}\lambda\nonumber\\
&\leq\frac{(1+2s)\mu}{1-4s\mu}\frac{2\|x_{-\max(s)}\|_1}{\sqrt{s}}+\frac{15\sqrt{s}}{16\mu(1-4s\mu)}\lambda.
\end{align}

Now, we estimate $\|h_{-\max(s)}\|_2$. Using (\ref{LassoConeconstriantinequality2}), we have
\begin{align}\label{e2.7}
&\|h_{-\max(s)}\|_2\leq\sqrt{\|h_{-max(s)}\|_1\|h_{-\max(s)}\|_{\infty}}\nonumber\\
&\leq\sqrt{\big(3\|h_{\max(s)}\|_1+4\|x_{-\max(s)}\|_1\big)\frac{\|h_{\max(s)}\|_1}{s}}\nonumber\\
&\leq\sqrt{\frac{3(\sqrt{s}\|h_{\max(s)}\|_2)^2}{s}+\frac{4\|x_{-\max(s)}\|_1(\sqrt{s}\|h_{\max(s)}\|_2)}{s}}\nonumber\\
&=\sqrt{3\|h_{\max(s)}\|_2^2+\frac{4\|x_{-\max(s)}\|_1}{\sqrt{s}}\|h_{\max(s)}\|_2}.
\end{align}

It follows from (\ref{e2.7}) that
\begin{align}
\|h\|_2&\leq\sqrt{\|h_{\max(s)}\|_2^2+3\|h_{\max(s)}\|_2^2+\frac{4\|x_{-\max(s)}\|_1}{\sqrt{s}}\|h_{\max(s)}\|_2}\nonumber\\
&\leq2\|h_{\max(s)}\|_2+\frac{1}{2}\frac{2\|x_{-\max(s)}\|_1}{\sqrt{s}}.
\end{align}
Finally, by (\ref{e2.6}), we get
\begin{align*}
\|h\|_2&\leq\frac{15\sqrt{s}}{8\mu(1-4s\mu)}\lambda+\Big(\frac{2(1+2s)\mu}{1-4s\mu}+\frac{1}{2}\Big)\frac{2\|x_{-\max(s)}\|_1}{\sqrt{s}},
\end{align*}
which finishes our proof.
\end{proof}

\subsection{Stable Recovery for QCBP and Dantzig Selector \label{s2.2}}
\hskip\parindent

In this subsection, we give out an exact expression of the stable recovery of QCBP (\ref{QCBPmodel}) and Dantzig selector model (\ref{DantzigselectorMedel}) for non-sparse signals $x$.

\begin{theorem}\label{DantzigselectorTheorem}
Assume the measurement matrix $A$ satisfies the MIP with
$$
\mu<\frac{1}{2s-1}.
$$
Let $\hat{x}^{DS}$ be the solution to the Dantzig selector (\ref{DantzigselectorMedel}), then
\begin{align*}
\|\hat{x}^{DS}-x\|_2\ \leq\frac{2\sqrt{2}\sqrt{s}}{1-(2s-1)\mu}\eta+\Big(\frac{\sqrt{2}s\mu}{1-(2s-1)\mu}+\frac{1}{2\sqrt{2}}\Big)\frac{2\|x_{-\max(s)}\|_1}{\sqrt{s}}.
\end{align*}
Let $\hat{x}^{\ell_2}$ be the solution to the QCBP (\ref{QCBPmodel}), then
\begin{align*}
\|\hat{x}^{\ell_2}-x\|_2\leq\frac{2\sqrt{2}\sqrt{1+(s-1)\mu}}{1-(2s-1)\mu}\eta
+\Big(\frac{\sqrt{2}s\mu}{1-(2s-1)\mu}+\frac{1}{2\sqrt{2}}\Big)\frac{2\|x_{-\max(s)}\|_1}{\sqrt{s}}.
\end{align*}
\end{theorem}

We need the following cone constraint inequality, which comes from \cite[Page 1215]{CRT2006} for $\hat{x}^{\ell_2}$, and \cite[Page 2330]{CT2007} for $\hat{x}^{DS}$.

\begin{lemma}\label{DSConeconstriantinequalityLemma}
The minimization solution $\hat{x}^{DS}$ of (\ref{DantzigselectorMedel}) and $\hat{x}^{\ell_2}$ of (\ref{QCBPmodel}) satisfy
\begin{align*}
\|h_{-\max(s)}\|_1\leq \|h_{\max(s)}\|_1+2\|x_{-\max(s)}\|_1,
\end{align*}
where $h=\hat{x}^{DS}-x$ or $h=\hat{x}^{\ell_2}-x$.
\end{lemma}

The proof is similar to the proofs of Theorems 2.1-2.2 in \cite{CWX2010} and we omit its details.

\subsection{Lower bounds of $\|\hat{x}-x\|_2^2$ \label{s2.3}}
\hskip\parindent

By Theorem \ref{LassoTheorem} and \cite[Theorem 2.2]{CWX2010}, there exist upper bounds for the error $\|\hat{x}^{L}-x\|_2^2$ and $\|\hat{x}^{DS}-x\|_2^2$ for sparse signals $x$. Whether there exist lower bounds for them? Cand\`{e}s and Plan \cite{CP2011} gave out lower bounds of them for Gaussian noises under the restricted isometry property, in the sense of expectation and probability.  In this subsection, we give an affirmative answer of this question in the framework of mutual incoherence property.

We consider the Gaussian noise model
\begin{align}\label{Gaussiannoise}
b=Ax+z,~~z\sim \mathcal{N}(0,\sigma^2I_m).
\end{align}
We shall assume that the noise level $\sigma$ is known.

For the Gaussian observations, we have the following probability inequality.
\begin{lemma}\cite[Lemma 5.1]{CXZ2009}\label{Probabilityinequality}
The Gaussian error $z\sim\mathcal{N}(0,\sigma^2 I_m)$ satisfies
$$
P\big(\|A^*z\|_{\infty}\leq\sigma\sqrt{2\log n}\big)\geq1-\frac{1}{2\sqrt{\pi\log n}}.
$$
\end{lemma}

By Theorem \ref{LassoTheorem}, \cite[Theorem 2.2]{CWX2010} and Lemma \ref{Probabilityinequality}, we can get upper bounds of $\|\hat{x}^{L}-x\|_2^2$ and $\|\hat{x}^{DS}-x\|_2^2$ for Gaussian noise observations as follows.

\begin{proposition}\label{Upperbound}
Let $z\sim\mathcal{N}(0,\sigma^2I_m)$.
\begin{itemize}
\item[(1)]
Assume the measurement matrix $A$ satisfies the MIP with
$$
\mu<\frac{1}{2s-1},
$$
and $x$ is $s$-sparse.  Then with probability at least
$$
1-\frac{1}{2\sqrt{\pi\log n}},
$$
$\hat{x}^{DS}$ satisfies
\begin{align*}
\|\hat{x}^{DS}-x\|_2^2\leq\frac{16\log n}{\big(1-(2s-1)\mu\big)^2}s\sigma^2.
\end{align*}
\item[(2)]
Assume the measurement matrix $A$ satisfies the MIP with
$$
\mu<\frac{1}{4s},
$$
and $x$ is $s$-sparse.  Then with probability at least
$$
1-\frac{1}{2\sqrt{\pi\log n}},
$$
$\hat{x}^{L}$ satisfies
\begin{align*}
\|\hat{x}^{L}-x\|_2^2\leq\frac{32\log n}{\big(\mu(1-4s\mu)\big)^2}s\sigma^2.
\end{align*}
\end{itemize}
\end{proposition}

We point out that, Proposition \ref{Upperbound} is nearly optimal in the sense that no estimator can do essentially better without further assumptions, as seen by lower-bounding the expected minimax error.

\begin{theorem}\label{ExpectionLowerbound}
Suppose that the measurement matrix $A$ is fixed and satisfies the MIP, and that $z\sim \mathcal{N}(0,\sigma^2 I_m)$. Then any estimator $\hat{x}$
obeys
\begin{align}
\sup_{x:\|x\|_0\leq s}\mathbb{E}\|\hat{x}-x\|_2^2\geq\frac{s\sigma^2}{1+(s-1)\mu}.
\end{align}
\end{theorem}

Note that this lower bound is in expectation, while the upper bound holds with high probability. To address this, we also prove the following complementary theorem.

\begin{theorem}\label{ProbabilityLowerbound}
Suppose that the measurement matrix $A$ is fixed and satisfies the MIP, and that $z\sim \mathcal{N}(0,\sigma^2 I_m)$. Then any estimator $\hat{x}$ obeys
\begin{align*}
\sup_{x:\|x\|_0\leq s}\mathbb{P}\Big(\|\hat{x}-x\|_2^2\geq\frac{ns\sigma^2}{2(1+(s-1)\mu)}\Big)\geq1-e^{-ns/16}.
\end{align*}
\end{theorem}

The proofs need the following two vital lemmas.

\begin{lemma}\cite[Page 403]{LC1998}\label{ExpectionLowerboundLemma}
Let $\lambda_j(A^*A)$ be the eigenvalues of the matrix $A^*A$, then
$$
\inf_{\hat{x}}\sup_{x\in\mathbb{R}^n}\mathbb{E}\|\hat{x}-x\|_2^2=\sigma^2\text{trace}((A^*A)^{-1})=\sum_{j}\frac{\sigma^2}{\lambda_j(A^*A)}.
$$
In particular, if one of the eigenvalues vanishes (as in the case in which $m<n$), then the minimax risk is unbounded.
\end{lemma}

\begin{lemma}\cite{CP2011}\label{ProbabilityLowerboundLemma}
Consider Gaussian model (\ref{Gaussiannoise}), then
\begin{align*}
\inf_{\hat{x}}\sup_{x\in\mathbb{R}^n}\mathbb{P}\bigg(\|\hat{x}-x\|_2^2\geq\frac{n\sigma^2}{2\|A\|_{2\rightarrow2}^2}\bigg)\geq1-e^{-n/16}.
\end{align*}
\end{lemma}

We are now in position to prove Theorem \ref{ExpectionLowerbound} and Theorem \ref{ProbabilityLowerbound}.
\begin{proof}[Proof of Theorem \ref{ExpectionLowerbound}]
 By
$$
\|\hat{x}-x\|_2^2=\|\hat{x}_{\text{supp}(x)}-x_{\text{supp}(x)}\|_2^2+\|\hat{x}_{(\text{supp}(x))^c}\|_2^2
\geq\|\hat{x}_{\text{supp}(x)}-x_{\text{supp}(x)}\|_2^2,
$$
we have
\begin{align*}
\inf_{\hat{x}}\sup_{x:\|x\|_0\leq s}\mathbb{E}\|\hat{x}-x\|_2^2
&\geq\inf_{\hat{x}}\sup_{x: \|x\|_0\leq s}\mathbb{E}\|\hat{x}_{\text{supp}(x)}-x_{\text{supp}(x)}\|_2^2\\
&=\sum_{j}\frac{\sigma^2}{\lambda_j(A_{\text{supp}(x)}^*A_{\text{supp}(x)})},
\end{align*}
where the equality follows from Lemma \ref{ExpectionLowerboundLemma}. We give an upper bound of $\lambda_{j}(A_{\text{supp}(x)}^*A_{\text{supp}(x)})$ as follows:
\begin{align*}
\lambda_{\max}(A_{\text{supp}(x)}^*A_{\text{supp}(x)})&=\sup_{\|u\|_2\leq 1}\langle u, A_{\text{supp}(x)}^*A_{\text{supp}(x)}u \rangle
=\sup_{\|u\|_2\leq 1}\langle A_{\text{supp}(x)}u, A_{\text{supp}(x)}u \rangle\\
&=\sup_{\|u\|_2\leq 1}\|A_{\text{supp}(x)}u\|_2^2\leq1+(s-1)\mu,
\end{align*}
where the inequality follows from Lemma \ref{MIPLemma}. Therefore,
\begin{align*}
\inf_{\hat{x}}\sup_{x:\|x\|_0\leq s}\mathbb{E}\|\hat{x}-x\|_2^2
&\geq\sum_{j}\frac{\sigma^2}{1+(s-1)\mu}=\frac{s\sigma^2}{1+(s-1)\mu}.
\end{align*}
\end{proof}

\begin{proof}[Proof of Theorem \ref{ProbabilityLowerbound}]
By
$$
\|A_{\text{supp}(x)}\|_{2\rightarrow 2}^2=\lambda_{\max}(A_{\text{supp}(x)}^*A_{\text{supp}(x)})\leq1+(s-1)\mu,
$$
we have
\begin{align*}
\inf_{\hat{x}}&\sup_{x:\|x\|_0\leq s}\mathbb{P}\Big(\|\hat{x}-x\|_2^2\geq\frac{ns\sigma^2}{2(1+(s-1)\mu)}\Big)\\
&\geq\inf_{\hat{x}}\sup_{x:\|x\|_0\leq s}\mathbb{P}\Big(\|\hat{x}_{\text{supp}(x)}-x_{\text{supp}(x)}\|_2^2
\geq\frac{ns\sigma^2}{2(1+(s-1)\mu)}\Big)\\
&\geq\inf_{\hat{x}}\sup_{x:\|x\|_0\leq s}\mathbb{P}\bigg(\|\hat{x}_{\text{supp}(x)}-x_{\text{supp}(x)}\|_2^2
\geq\frac{ns\sigma^2}{2\|A_{\text{supp}(x)}\|_{2\rightarrow 2}}\bigg)\\
&\geq1-e^{-ns/16},
\end{align*}
where the last inequality follows from Lemma \ref{ProbabilityLowerboundLemma}.
\end{proof}

%%%%%%%%%%%%%%%%%%%%%%%%%%%%%%%%%%%%%%%%%%%%%%%%%%%%%%%%%%%%%%%
%%%%%%%%%%%%%%%%%%%%%%%% Section 3  %%%%%%%%%%%%%%%%%%%%%%%%%%%
%%%%%%%%%%%%%%%%%%%%%%%%%%%%%%%%%%%%%%%%%%%%%%%%%%%%%%%%%%%%%%%

\section{Oracle Inequality \label{s3}}
\hskip\parindent

The oracle inequality approach was introduced by Donoho
and Johnstone \cite{DJ1994} in the context of wavelet thresholding for signal denoising. It provides an effective tool for studying the performance of an estimation procedure by comparing it to that of an ideal estimator. This approach has been extended to study compressed sensing by Cand\`{e}s and Tao's groundbreaking work \cite{CT2007}.  In \cite{CT2007}, they developed an oracle inequality for Dantzig selector $\hat{x}^{DS}$ in the Gaussian noise setting in the framework of restricted isometry property. Later, Cand\`{e}s and Plan \cite{CP2011} extended it to matrix Lasso and matrix Dantzig selector under the condition of restricted isometry property, for both low-rank matrix and not low-rank matrix. And almost in the same time, Cai, Wang and Xu \cite{CWX2010} extended it to Dantzig selector $\hat{x}^{DS}$ for sparse signals in the framework of mutual incoherence property. Motivated by \cite{CP2011} and \cite{CWX2010}, we consider oracle inequality via Lasso for both sparse signals and non-sparse signals under the framework of mutual incoherence property.

Before stating our main results,  we first give two notations. Let
$$
S_0=\{j\in\mathbb{R}^n: |x(j)|\geq\sigma\}
$$
and
\begin{align}\label{e3.2}
K(\xi,x)=\sigma^2\|\xi\|_0+\|x-\xi\|_2^2.
\end{align}
Note that
\begin{align}\label{e3.3}
K(x_{S_0},x)&=\sigma^2\|x_{S_0}\|_0+\|x_{S_0^c}\|_2^2=\sigma^2\sum_{j\in S_0}\chi_{S_0}(j)+\sum_{j\in S_0^c}|x(j)|^2\nonumber\\
&=\sum_{j}\min\{\sigma^2,|x(j)|^2\}=\sigma^2\sum_{j}\min\{1,\frac{|x(j)|^2}{\sigma^2}\}=:\sigma^2\tau.
\end{align}

%%%%%%%%%%%%%%%%%%%%%%%%%%%%%%%%%%%%%%%%%%%%%%%%%%%%%%%%%%%%%%%
\subsection{Sparse Vector Case \label{s3.1}}
\hskip\parindent

The oracle inequality of Dantzig selector for sparse vector under the condition MIP, which was obtained by Cai, Wang and Xu in \cite{CWX2010}, can be stated as follows.

\begin{theorem}\label{OracleSparseDS}
Consider the Gaussian noise model (\ref{Gaussiannoise}). Suppose $x$ is $s$ sparse and measurement matrix $A\in\mathbb{R}^{m\times n}$ satisfies MIP with
$$
\mu<\frac{1}{2s-1}.
$$
Set $\eta^*=\sigma(3/2+\sqrt{2\log n})$. Let $\hat{x}^{DS}$ be the minimizer of the problem
\begin{align*}
\min_{y\in\mathbb{R}^n}~\|y\|_1 ~~\text{subject~ to}~~\|A^*(b-Ay)\|_{\infty}\leq\eta^*.
\end{align*}
Then with probability at least
$$
1-\frac{1}{2\sqrt{\pi\log n}},
$$
$\hat{x}^{DS}$ satisfies
\begin{align*}
\|\hat{x}^{DS}-x\|_2^2&\leq\frac{8\big(2+\sqrt{2\log n}\big)^2}{\big(1-(2s-1)\mu\big)^2}\sum_{j}\min\{\sigma^2,|x(j)|^2\}\\
&\leq\frac{8\big(2+\sqrt{2\log n}\big)^2}{\big(1-(2s-1)\mu\big)^2}\Big(\sigma^2+\sum_{j}\min\{\sigma^2,|x(j)|^2\}\Big).
\end{align*}
\end{theorem}

Now, we consider the oracle inequality of Lasso for sparse vector under the condition MIP.

\begin{theorem}\label{OracleSparseLasso}
Consider the Gaussian noise model (\ref{Gaussiannoise}). Suppose that $x$ is $s$-sparse and measurement matrix $A\in\mathbb{R}^{m\times n}$ satisfies MIP with
$$\mu<\frac{1}{4s}.$$
Set $\lambda^*=2\sigma(\frac{5}{4}+\sqrt{2\log n})$. Let $\hat{x}^{L}$ be the minimizer of
\begin{align}\label{SparseGaussiannoiseLassoModel}
\min_{y\in\mathbb{R}^n}~\lambda^*\|y\|_1+\frac{1}{2}\|Ax-b\|_2^2.
\end{align}
Then with probability at least
$$
1-\frac{1}{2\sqrt{\pi\log n}},
$$
$\hat{x}^{L}$ satisfies
\begin{align}\label{e3.1}
\|\hat{x}^{L}-x\|_2^2\leq\frac{16\big(2+\sqrt{2\log n}\big)^2}{\big(\mu(1-4s\mu)\big)^2}\sum_{j}\min\{\sigma^2,|x(j)|^2\}.
\end{align}
\end{theorem}

\begin{proof}
Without loss of generality, we assume $\text{supp}(x)\subset S$ with $|S|=s$.
Set $\lambda=\sigma\sqrt{2\log n}$. By Lemma \ref{Probabilityinequality}, event $E=\{z\in\mathbb{R}^m: \|A^*z\|_{\infty}\leq\lambda\}$ occurs with probability at least
$$
1-\frac{1}{2\sqrt{\pi\log n}}.
$$
In the following, we shall assume that event $E$ occurs.

By the definition of $K(\xi,x)$, we have
$$
K(x_{S_0},x)\leq K(x,x),
$$
which implies that $\|x_{S_0}\|_0\leq\tau\leq\|x\|_0\leq s$.

Write $x^{(1)}=x_{S_0}$ and $x^{(2)}=x_{S\backslash S_0}$, then we have $x=x_S=x^{(1)}+x^{(2)}$ and $|x(j)|<\sigma$ for $j\in S\backslash S_0$.
Therefore
\begin{align}\label{e3.4}
\|x^{(2)}\|_1&=\sum_{j\in S\backslash S_0}|x(j)|<\sigma|S\backslash S_0|,
\end{align}
and
\begin{align}\label{e3.5}
\|x^{(2)}\|_2&=\sqrt{\sum_{j\in S\backslash S_0}\min\{\sigma^2,|x(j)|^2\}}\leq\sigma\sqrt{\tau}.
\end{align}

Next, we verify that $x^{(1)}$ satisfies $\|A^*(b-Ax^{(1)})\|_{\infty}\leq\lambda^*/2$. In fact, for any $1\leq i\leq n$,
\begin{align*}
|\langle A_i,Ax^{(1)}-b \rangle|&=|\langle A_i, Ax-b\rangle-\langle A_i, Ax^{(2)}\rangle|\leq|\langle A_i, z\rangle|+|\langle A_i, Ax^{(2)}\rangle|\\
&\leq\lambda+\sum_{j\in S\backslash S_0}|\langle A_i,A_j\rangle||x^{(2)}(j)|\\
&\leq
\begin{cases}
\lambda+\mu\|x^{(2)}\|_1, &i\in S_0\cup S^c\\
\lambda+\mu\|x^{(2)}\|_1+(1-\mu)|x^{(2)}(i)|,&i\in S\backslash S_0
\end{cases}\\
&\leq
\begin{cases}
\lambda+\mu\sigma|S\backslash S_0|, &i\in S_0\cup S^c\\
\lambda+\mu\sigma|S\backslash S_0|+(1-\mu)\sigma,&i\in S\backslash S_0
\end{cases}\\
&\leq\lambda+\mu\sigma|S\backslash S_0|+(1-\mu)\sigma\\
&=\lambda+\sigma+\sigma\mu(|S|-|S_0|-1)\\
&\leq\lambda+\sigma+\sigma\mu s,
\end{align*}
then it follows from $\mu<1/(4s)$ that
\begin{align*}
|\langle A_i,Ax^{(1)}-b \rangle|
&\leq\lambda+\sigma+\frac{\sigma}{4}=\sigma\Big(\frac{5}{4}+\sqrt{2\log n}\Big)=\frac{\lambda^*}{2},
\end{align*}
then
$$
\|A^*(b-Ax^{(1)})\|_{\infty}\leq\frac{\lambda^*}{2}.
$$

Therefore, from Theorem \ref{LassoTheorem} and $x^{(1)}$ is $|S_0|$-sparse, we have
\begin{align*}
\|\hat{x}^{L}-x^{(1)}\|_2&\leq\frac{15\sqrt{|S_0|}}{8\mu(1-4|S_0|\mu)}\lambda^*
=\frac{15\sigma\sqrt{|S_0|}}{4\mu(1-4|S_0|\mu)}\Big(\frac{5}{4}+\sqrt{2\log n}\Big)\nonumber\\
&\leq\frac{15\sigma\sqrt{\tau}}{4\mu(1-4\tau\mu)}\Big(\frac{5}{4}+\sqrt{2\log n}\Big)\nonumber\\
&=\frac{15\sigma\sqrt{\tau}}{4\mu(1-4\tau\mu)}\big(2+\sqrt{2\log n}\big)-\frac{45}{16\mu(1-4\tau\mu)}\sigma\sqrt{\tau}.
\end{align*}
By
$$
\mu(1-4\tau\mu)=-4\tau(\mu-\frac{1}{8\tau})^2+\frac{1}{16\tau}\leq\frac{1}{16\tau}\leq\frac{1}{16},
$$
we have
$$
\frac{45}{16\mu(1-4\tau\mu)}\sigma\sqrt{\tau}\geq\frac{45}{16}16\sigma\sqrt{\tau}\geq\sigma\sqrt{\tau}.
$$
Therefore
\begin{align}\label{e3.6}
\|\hat{x}^{L}-x^{(1)}\|_2
&\leq\frac{15\sigma\sqrt{\tau}}{4\mu(1-4\tau\mu)}\big(2+\sqrt{2\log n}\big)-\sigma\sqrt{\tau}.
\end{align}
Combination of (\ref{e3.5}) and (\ref{e3.6}) yields
\begin{align*}
&\|\hat{x}^{L}-x\|_2\leq\|\hat{x}^{L}-x^{(1)}\|_2+\|x^{(2)}\|_2\\
&\leq\frac{15\sigma\sqrt{\tau}}{4\mu(1-4\tau\mu)}\big(2+\sqrt{2\log n}\big)-\sigma\sqrt{\tau}+\sigma\sqrt{\tau}\\
&\leq\frac{15}{4\mu(1-4\tau\mu)}\big(2+\sqrt{2\log n}\big)\sigma\sqrt{\tau}.
\end{align*}
Consequently,
\begin{align*}
\|\hat{x}^{L}-x\|_2^2&\leq\frac{16\big(2+\sqrt{2\log n}\big)^2}{\big(\mu(1-4s\mu)\big)^2}\sigma^2\tau
=\frac{16\big(2+\sqrt{2\log n}\big)^2}{\big(\mu(1-4s\mu)\big)^2}\sum_{j}\min\{\sigma^2,|x(j)|^2\},
\end{align*}
where the equality follows by (\ref{e3.3}).
\end{proof}

\subsection{Extension to General Vector Case \label{s3.2}}
\hskip\parindent

In this subsection, we demonstrate the error bound when $x$ is non-sparse.

\begin{theorem}\label{OracleGeneral}
Consider the Gaussian noise model (\ref{Gaussiannoise}). Suppose that $A$ is sampled from the Gaussian measurement ensemble and $s^*\leq m/{\log(en/m)}$.
\begin{itemize}
\item[(1)]
Suppose measurement matrix $A\in\mathbb{R}^{m\times n}$ satisfies
    $$
    \mu<\frac{1}{2s^*-1}.
    $$
   Set $\eta^*=\sigma(3/2+\sqrt{2\log n})$. Let $\hat{x}^{DS}$ be the minimizer of the problem
   \begin{align}\label{GaussiannoiseDSModel}
   \min_{y\in\mathbb{R}^n}~\|y\|_1 ~~\text{subject~ to}~~\|A^*(b-Ay)\|_{\infty}\leq\eta^*.
   \end{align}
   Then with probability at least
    $$
     1-e^{-m/100}-\frac{1}{2\sqrt{\pi\log n}},
    $$
    $\hat{x}^{DS}$ satisfies
\begin{align*}
\|\hat{x}^{DS}-x\|_2^2
&\leq6\bigg(\frac{\big(138-34(4s^*-1)\mu\big)\sqrt{1+(s^*-1)\mu}}{1-(2s^*-1)\mu}\bigg)^2\big(2+\sqrt{2\log n}\big)^2\\
&\hspace*{12pt}\times\bigg(\sum_{j\in \text{supp}(x_{\max(s^*)})}\min\{\sigma^2,|x(j)|^2\}+\|x_{-\max(s^*)}\|_2^2\bigg).
\end{align*}
\item[(2)]Suppose measurement matrix $A\in\mathbb{R}^{m\times n}$ satisfies
    $$
    \mu<\frac{1}{4s^*}.
    $$
    Set $\lambda^*=2\sigma(5/4+\sqrt{2\log n})$. Let $\hat{x}^{L}$ be the minimizer of
   \begin{align}\label{GaussiannoiseLassoModel}
   \min_{y\in\mathbb{R}^n}~\lambda^*\|y\|_1+\frac{1}{2}\|Ay-b\|_2^2.
   \end{align}
    Then with probability at least
    $$
     1-e^{-m/100}-\frac{1}{2\sqrt{\pi\log n}},
    $$
    $\hat{x}^{L}$ satisfies
\begin{align*}
\|\hat{x}^{L}-x\|_2^2&
\leq24\Bigg(\frac{\Big(2+34\mu\big(4-3(s^*-1)\mu\big)\Big)\sqrt{1+(s^*-1)\mu}}{\mu(1-4s^*\mu)}\Bigg)^2\Big(\frac{5}{4}+\sqrt{2\log n}\Big)^2\nonumber\\
&\hspace*{12pt}\times\bigg(\sum_{j\in \text{supp}(x_{\max(s^*)})}\min\{\sigma^2,|x(j)|^2\}+\|x_{-\max(s^*)}\|_2^2\bigg).
\end{align*}
\end{itemize}
\end{theorem}

\begin{remark}
In \cite{SV2008}, Schnass and Vandergheynst showed that the coherence of a matrix $A\in\mathbb{R}^{m\times n}$ with $\ell_2$-normalized columns satisfies
$$
\mu\geq \sqrt{\frac{n-m}{m(n-1)}}.
$$
And for large $n$, $\mu\sim1/\sqrt{m}$. Therefore, $1/\sqrt{m}<1/(2s^*-1)$ or $1/\sqrt{m}<1/(4s^*)$ leads to $m\geq (s^*)^2$.
\end{remark}

Two useful results (Lemma \ref{Highnoiselevel} and Proposition \ref{Lownoiselevel}) are established, in order to prove Theorem \ref{OracleGeneral}. The first one is used in the high noise level case.

\begin{lemma}\label{Highnoiselevel}
Let $\bar{x}=\arg\min_{\xi\in\mathbb{R}^n} K(\xi,x)$ and set $\bar{s}=\max\{\|x_{S_0}\|_0,\|\bar{x}\|_0\}$.
\begin{itemize}
\item[(1)] If $\mu<1/(2\bar{s}-1)$, then the solution $\hat{x}^{DS}$ to (\ref{GaussiannoiseDSModel}) satisfies
\begin{align*}
\|\hat{x}^{DS}-x\|_2^2\leq\frac{16(2+\sqrt{2\log n})^2}{\big(1-(2\bar{s}-1)\mu\big)^2}\sum_{j}\min\{\sigma^2,|x(j)|^2\}
\end{align*}
with probability at least
$$
1-\frac{1}{2\sqrt{\pi\log n}}.
$$
\item[(2)]
If $\mu<1/(4\bar{s})$, then the solution $\hat{x}^{L}$ to (\ref{GaussiannoiseLassoModel}) satisfies
\begin{align*}
\|\hat{x}^{L}-x\|_2^2\leq\frac{32(2+\sqrt{2\log n})^2}{\big(\mu(1-4\bar{s}\mu)\big)^2}\sum_{j}\min\{\sigma^2,|x(j)|^2\}
\end{align*}
with probability at least
$$
1-\frac{1}{2\sqrt{\pi\log n}}.
$$
\end{itemize}
\end{lemma}
\begin{proof}
First, we consider $\hat{x}^L$. We can rewrite
$$
\|\hat{x}^{L}-x\|_2^2\leq2\|\hat{x}^{L}-\bar{x}\|_2^2+2\|\bar{x}-x\|_2^2.
$$

We bound $\|\hat{x}^{L}-\bar{x}\|_2^2$ using the exact same steps as in the proof of Theorem \ref{OracleSparseLasso}, and obtain
\begin{align}\label{e3.9}
\|\hat{x}^{L}-\bar{x}\|_2^2&\leq\frac{16(2+\sqrt{2\log n})^2}{\big(\mu(1-4\|\bar{x}\|_0\mu)\big)^2}\big(\sigma\sqrt{\|\bar{x}\|_0}\big)^2
\leq\frac{16(2+\sqrt{2\log n})^2}{\big(\mu(1-4\bar{s}\mu)\big)^2}\sigma^2\|\bar{x}\|_0.
\end{align}
Hence
\begin{align*}
\|\hat{x}^{L}-x\|_2^2&\leq\frac{32(2+\sqrt{2\log n})^2}{\big(\mu(1-4\bar{s}\mu)\big)^2}\sigma^2\|\bar{x}\|_0+2\|\bar{x}-x\|_2^2\\
&\leq\frac{32(2+\sqrt{2\log n})^2}{\big(\mu(1-4\bar{s}\mu)\big)^2}(\sigma^2\|\bar{x}\|_0+\|x-\bar{x}\|_2^2)\\
&=\frac{32(2+\sqrt{2\log n})^2}{\big(\mu(1-4\bar{s}\mu)\big)^2}K(\bar{x},x).
\end{align*}
By the definition of $\bar{x}$, we know that $K(\bar{x},x)\leq K(x_{S_0},x)$. Therefore
\begin{align*}
\|\hat{x}^{L}-x\|_2^2&\leq\frac{32(2+\sqrt{2\log n})^2}{\big(\mu(1-4\bar{s}\mu)\big)^2}K(x_{S_0},x)
\leq\frac{32(2+\sqrt{2\log n})^2}{\big(\mu(1-4\bar{s}\mu)\big)^2}\sum_{j}\min\{\sigma^2,|x(j)|^2\},
\end{align*}
where the last steps follows from (\ref{e3.3}).

For the solution $\hat{x}^{DS}$, we can replace (\ref{e3.9}) by
\begin{align}\label{e3.9-1}
\|\hat{x}^{DS}-\bar{x}\|_2^2&\leq\frac{8(2+\sqrt{2\log n})^2}{\big(1-(2\|\bar{x}\|_0-1)\mu\big)^2}\big(\sigma\sqrt{\|\bar{x}\|_0}\big)^2
\leq\frac{8(2+\sqrt{2\log n})^2}{\big(1-(2\bar{s}-1)\mu\big)^2}\sigma^2\|\bar{x}\|_0,
\end{align}
where the first inequality follows by Theorem \ref{OracleSparseDS}.
Then by similar proof above, we can get
\begin{align*}
\|\hat{x}^{DS}-x\|_2^2\leq\frac{16(2+\sqrt{2\log n})^2}{\big(1-(2\bar{s}-1)\mu\big)^2}\sum_{j}\min\{\sigma^2,|x(j)|^2\}.
\end{align*}
\end{proof}

In order to provide theoretical error bounds when the noise level is low, a useful property of measurement matrix is
needed. We call it the LQ property or quotient property, which was introduced by Wojtaszczyk in \cite{W2010}. And in \cite[Chapter 11]{FR2013} and \cite{F2014}, Foucart and Rauhut also investigated this property.

\begin{definition}(\cite{W2010,FR2013,F2014})
Given $q\geq 1$, a measurement matrix $A\in\mathbb{R}^{m\times n}$ is said  to possess the $\ell_q$ quotient property with constant $\alpha>0$ relative to a norm $\|\cdot\|_2$ on $\mathbb{R}^n$, if for all $x\in\mathbb{R}^n$, there exists $\tilde{x}\in\mathbb{R}^n$ such that
\begin{align*}
Ax=A\tilde{x}~\text{and}~\|\tilde{x}\|_q&\leq \frac{1}{\alpha}\|Ax\|_2.
\end{align*}
If $q=1$, we denote it as $LQ(\alpha)$.
\end{definition}

Wojtaszczyk \cite{W2010} showed that Gaussian random matrix satisfies this property with high probability, save for the extra requirement that $n\geq cm(\log m)^{\varsigma}$ for some $\varsigma>0$. And Foucart and Rauhut \cite[Chapter 11]{FR2013} and \cite{F2014} weakened this requirement to $n\geq 2m$.

\begin{lemma}\cite{FR2013}\label{LQforGaussian}
For $n\geq 2m$, if $A$ is a draw of an $m\times n$ Gaussian random matrix, then the matrix $\tilde{A}=A/\sqrt{m}$ possesses the $LQ(\alpha)$ with constant $\alpha=1/(34\sqrt{s_*})$ with probability at least $1-e^{-m/100}$, where $s_*=m/\log(en/m)$.
\end{lemma}

Using the LQ property, we can now bound the error when the noise level is low.

\begin{proposition}\label{Lownoiselevel}
Suppose that $A$ satisfies $LQ(\alpha)$ with constant $\alpha=\sqrt{m}/(34\sqrt{s_*})$.
\begin{itemize}
\item[(1)]
If $A$ also satisfies MIP with
$$
\mu<\frac{1}{4s^*},
$$
$\|A^*z\|_{\infty}\leq\lambda^*/2$, then the solution $\hat{x}^{L}$  to (\ref{GaussiannoiseLassoModel}) satisfies
\begin{align*}
\|\hat{x}^{L}-x\|_2\leq\frac{\Big(2+34\mu\big(4-3(s^*-1)\big)\Big)\sqrt{1+(s^*-1)\mu}}{\mu(1-4s^*\mu)}\big(\sqrt{s^*}\lambda^*+\|x_{-\max(s^*)}\|_2\big).
\end{align*}
\item[(2)]
If $A$ also satisfies MIP with
$$
\mu<\frac{1}{2s^*-1},
$$
then the solution $\hat{x}^{DS}$ to (\ref{GaussiannoiseDSModel}) satisfies
\begin{align*}
\|\hat{x}^{DS}-x\|_2
\leq\frac{\big(138-34(4s^*-1)\mu\big)\sqrt{1+(s^*-1)\mu}}{1-(2s^*-1)\mu}\big(\|x_{-\max(s^*)}\|_2+\sqrt{s^*}\eta^*\big).
\end{align*}
\end{itemize}
\end{proposition}

The following lemma, which is similar to \cite[Lemma 3.2]{W2010}, is useful in the proof of Proposition \ref{Lownoiselevel}.
\begin{lemma}\label{MIP-LQ}
For any $\theta>0$, suppose that matrix $A\in\mathbb{R}^{m\times n}$ satisfies $LQ(\theta/{\sqrt{s}})$ and  MIP with $\mu<1/(s-1)$, then for any $x\in\mathbb{R}^n$, there exists $\tilde{x}\in\mathbb{R}^n$ such that
\begin{align*}
\|\tilde{x}\|_2\leq\Big(\frac{1}{\theta}+\frac{\sqrt{1+(s-1)\mu}}{\theta(\sqrt{1-(s-1)\mu})}+\frac{1}{\sqrt{1-(s-1)\mu}}\Big)\|Ax\|_2.
\end{align*}
\end{lemma}
\begin{proof}
By $LQ$ property, we know that there exists $\tilde{x}\in\mathbb{R}^n$ such that
$$
Ax=A\tilde{x}~~\text{and}~~\|\tilde{x}\|_1\leq\frac{\sqrt{s}}{\theta}\|Ax\|_2.
$$
We decompose $\tilde{x}$ as
$$
\tilde{x}=\sum_{j\geq1}\tilde{x}_{S_j},
$$
where $S_1$ is the index set of the $s$ largest entries of $\tilde{x}$, and $S_2$ is the index set of the $s$ largest entries of $\tilde{x}_{S_1^c}$, and so on. The last index set may contain less $s$ elements.
Clearly, we have that for $j\geq2$
\begin{align*}
\|\tilde{x}_{S_{j}}\|_2&\leq\frac{\|\tilde{x}_{S_{j-1}}\|_1}{\sqrt{s}}.
\end{align*}
Therefore
\begin{align}\label{e3.10}
\sum_{j\geq 2}\|\tilde{x}_{S_j}\|_2\leq\sum_{j\geq 1}\frac{\|\tilde{x}_{S_{j}}\|_1}{\sqrt{s}}
=\frac{\|\tilde{x}_{S_1^c}\|_1}{\sqrt{s}}
\leq\frac{1}{\sqrt{s}}\frac{\sqrt{s}\|Ax\|_2}{\theta}=\frac{\|Ax\|_2}{\theta}.
\end{align}

In order to estimate $\|\tilde{x}\|_2$, we rewrite it as
$$
\|\tilde{x}\|_2\leq\|\tilde{x}_{S_1}\|_2+\|\tilde{x}_{S_1^c}\|_2.
$$

First, we deal with $\|\tilde{x}_{S_1^c}\|_2$. From (\ref{e3.10}), we have
\begin{align}\label{e3.11}
\|\tilde{x}_{S_1^c}\|_2&\leq\sum_{j\geq 2}\|\tilde{x}_{S_j}\|_2\leq\frac{\|Ax\|_2}{\theta}.
\end{align}

On the other hand, using the MIP condition, we get
\begin{align*}
\|\tilde{x}_{S_1}\|_2&\leq\frac{1}{\sqrt{1-(s-1)\mu}}\|A\tilde{x}_{S_1}\|_2\leq\frac{1}{\sqrt{1-(s-1)\mu}}(\|A\tilde{x}\|_2+\|A\tilde{x}_{S_1^c}\|_2)\\
&=\frac{1}{\sqrt{1-(s-1)\mu}}(\|Ax\|_2+\|A\tilde{x}_{S_1^c}\|_2).
\end{align*}
It follows from MIP and (\ref{e3.11}) that
\begin{align*}
\|A\tilde{x}_{S_1^c}\|_2&\leq\sum_{j\geq2}\|A\tilde{x}_{S_j}\|_2\leq\sqrt{1+(s-1)\mu}\sum_{j\geq 2}\|\tilde{x}_{S_j}\|_2\\
&\leq\sqrt{1+(s-1)\mu}\frac{\|Ax\|_2}{\theta}.
\end{align*}
Therefore,
\begin{align}\label{e3.12}
\|\tilde{x}_{S_1}\|_2\leq\frac{1}{\sqrt{1-(s-1)\mu}}\Big(\frac{\sqrt{1+(s-1)\mu}}{\theta}+1\Big)\|Ax\|_2.
\end{align}

Combining (\ref{e3.11}) and (\ref{e3.12}), we have
\begin{align*}
\|\tilde{x}\|_2&\leq\frac{\|Ax\|_2}{\theta}+\frac{1}{\sqrt{1-(s-1)\mu}}\Big(\frac{\sqrt{1+(s-1)\mu}}{\theta}+1\Big)\|Ax\|_2\\
&=\Big(\frac{1}{\theta}+\frac{\sqrt{1+(s-1)\mu}}{\theta(\sqrt{1-(s-1)\mu})}+\frac{1}{\sqrt{1-(s-1)\mu}}\Big)\|Ax\|_2.
\end{align*}
\end{proof}

Now we can prove Proposition \ref{Lownoiselevel}.
\begin{proof}[Proof of Proposition \ref{Lownoiselevel}]
We split $x$ as
$$
x=x_{\max(s^*)}+x_{-\max(s^*)},
$$
then we can write $\|\hat{x}^{L}-x\|_2$ as
\begin{align*}
\|\hat{x}^{L}-x\|_2\leq\|\hat{x}^{L}-x_{\max(s^*)}\|_2+\|x_{-\max(s^*)}\|_2.
\end{align*}
The condition $LQ(\sqrt{m}/(34\sqrt{s^*})$ implies that there exists $\tilde{x}\in\mathbb{R}^n$ such that
\begin{align}\label{e3.13}
Ax_{-\max(s^*)}=A\tilde{x}~~\text{and}~~\|\tilde{x}\|_1&\leq \frac{34\sqrt{s^*}}{\sqrt{m}}\|Ax_{-\max(s^*)}\|_2.
\end{align}
Therefore, we can rewrite $\|\hat{x}^{L}-x\|_2$ as
\begin{align*}
\|\hat{x}^{L}-x\|_2\leq\|\hat{x}^{L}-(x_{\max(s^*)}+\tilde{x})\|_2+\|\tilde{x}\|_2+\|x_{-\max(s^*)}\|_2.
\end{align*}

First, we estimate $\|\tilde{x}\|_2$. By Lemma \ref{MIP-LQ}, we have
\begin{align*}
\|\tilde{x}\|_2\leq\Big(\frac{34}{\sqrt{m}}+\frac{34\sqrt{1+(s^*-1)\mu}}{\sqrt{m}\sqrt{1-(s^*-1)\mu}}
+\frac{1}{\sqrt{1-(s^*-1)\mu}}\Big)\|Ax_{-\max(s^*)}\|_2.
\end{align*}
We decompose $x_{-\max(s^*)}$ as
$$
x_{-\max(s^*)}=\sum_{j\geq1}x_{S_j},
$$
where $S_1$ is the index set of $s^*$ largest entries of $x_{-\max(s^*)}$, and $S_2$ is the index set of $s^*$ largest entries of $x_{(\max(s^*)\cup S_1)^c}$, and so on. The last index set may contain less $s^*$ elements.
Then by MIP condition, we have
\begin{align}\label{e3.14}
\|Ax_{-\max(s^*)}\|_2^2&=\sum_{j\geq1}\|Ax_{S_j}\|_2^2\leq\big(1+(s^*-1)\mu\big)\sum_{j\geq1}\|x_{S_j}\|_2^2\nonumber\\
&=\big(1+(s^*-1)\mu\big)\|x_{-max(s^*)}\|_2^2.
\end{align}
Therefore,
\begin{align}\label{e3.15}
\|\tilde{x}\|_2\leq\Big(\frac{34}{\sqrt{m}}+\frac{34\sqrt{1+(s^*-1)\mu}}{\sqrt{m}(\sqrt{1-(s^*-1)\mu})}+\frac{1}{\sqrt{1-(s^*-1)\mu}}\Big)
\sqrt{1+(s^*-1)\mu}\|x_{-max(s^*)}\|_2.
\end{align}

It remains to estimate the term $\|\hat{x}^{L}-(x_{\max(s^*)}+\tilde{x})\|_2$. By Theorem \ref{LassoTheorem}, we have that when $\mu<1/(4s^*)$,
\begin{align*}
\|\hat{x}^{L}-(x_{\max(s^*)}+\tilde{x})\|_2
&\leq \frac{15\sqrt{s^*}}{8\mu(1-4s^*\mu)}\lambda^*
+\bigg(\frac{2(1+2s^*)\mu}{1-4s^*\mu}+\frac{1}{2}\bigg)\frac{2\|\tilde{x}\|_1}{\sqrt{s^*}}\\
&\leq \frac{2}{\mu(1-4s^*\mu)}\sqrt{s^*}\lambda^*+\frac{1+2\mu}{1-4s^*\mu}\frac{2\|\tilde{x}\|_1}{\sqrt{s^*}}.
\end{align*}
It then follows from (\ref{e3.13}) that
\begin{align}\label{e3.16}
\|\hat{x}^{L}&-(x_{\max(s^*)}+\tilde{x})\|_2\nonumber\\
&\leq \frac{2}{\mu(1-4s^*\mu)}\sqrt{s^*}\lambda^*
+\frac{1+2\mu}{1-4s^*\mu}\frac{68\sqrt{s^*}\|Ax_{-\max(s^*)}\|_2}{\sqrt{m}\sqrt{s^*}}\nonumber\\
&\leq\frac{2}{\mu(1-4s^*\mu)}\sqrt{s^*}\lambda^*+\frac{68(1+2\mu)\sqrt{1+(s^*-1)\mu}}{(1-4s^*\mu)\sqrt{m}}\|x_{-\max(s^*)}\|_2,
\end{align}
where the second inequality follows from (\ref{e3.14}).

Last, combination of (\ref{e3.15}) and (\ref{e3.16}) yields that
\begin{align*}
\|\hat{x}^{L}&-x\|_2\\
&\leq\Big(\frac{2}{\mu(1-4s^*\mu)}\sqrt{s^*}\lambda^*+\frac{68(1+2\mu)\sqrt{1+(s^*-1)\mu}}{(1-4s^*\mu)\sqrt{m}}\|x_{-\max(s^*)}\|_2\Big)\\
&\hspace*{12pt}+\Big(\frac{34}{\sqrt{m}}+\frac{34\sqrt{1+(s^*-1)\mu}}{\sqrt{m}(\sqrt{1-(s^*-1)\mu})}+\frac{1}{\sqrt{1-(s^*-1)\mu}}\Big)
\sqrt{1+(s^*-1)\mu}\|x_{-max(s^*)}\|_2\\
&\hspace*{12pt}+\|x_{-\max(s^*)}\|_2\\
&=\bigg(\Big(\frac{34}{\sqrt{m}}+\frac{34\sqrt{1+(s^*-1)\mu}}{\sqrt{m}(\sqrt{1-(s^*-1)\mu})}+\frac{1}{\sqrt{1-(s^*-1)\mu}}
+\frac{68(1+\mu)}{\sqrt{m}(1-4s^*\mu)}\Big)\\
&\hspace*{12pt}\times\sqrt{1+(s^*-1)\mu}+1\bigg)\|x_{-\max(s^*)}\|_2+\frac{2}{\mu(1-4s^*\mu)}\sqrt{s^*}\lambda^*\\\\
&\leq\bigg(\frac{34}{\sqrt{m}}\Big(1+\frac{\sqrt{1+(s^*-1)\mu}}{\sqrt{1-(s^*-1)\mu}}+\frac{2(1+2\mu)}{1-4s^*\mu}\Big)+\frac{2}{\mu(1-4s^*\mu)}\bigg)\\
&\hspace*{12pt}\times\sqrt{1+(s^*-1)\mu}\big(\sqrt{s^*}\lambda^*+\|x_{-\max(s^*)}\|_2\big)\\
&\leq\bigg(34\Big(1+\frac{1+(s^*-1)\mu}{1-4s^*\mu}+\frac{2(1+2\mu)}{1-4s^*\mu}\Big)+\frac{2}{\mu(1-4s^*\mu)}\bigg)\\
&\hspace*{12pt}\times\sqrt{1+(s^*-1)\mu}\big(\sqrt{s^*}\lambda^*+\|x_{-\max(s^*)}\|_2\big)\\
&=\frac{\Big(2+34\mu\big(4-3(s^*-1)\big)\Big)\sqrt{1+(s^*-1)\mu}}{\mu(1-4s^*\mu)}\big(\sqrt{s^*}\lambda^*+\|x_{-\max(s^*)}\|_2\big).
\end{align*}

For the $\hat{x}^{DS}$, by Theorem \ref{DantzigselectorTheorem}, we have that when $\mu<1/(2s^*-1)$,
\begin{align*}
\|\hat{x}^{DS}-(x_{\max(s^*)}+\tilde{x})\|_2&\leq\frac{2\sqrt{2}}{1-(2s^*-1)\mu}\sqrt{s^*}\eta^*
+\Big(\frac{\sqrt{2}s^*\mu}{1-(2s^*-1)\mu}+\frac{1}{2\sqrt{2}}\Big)\frac{2\|\tilde{x}\|_1}{\sqrt{s^*}}\\
&=\frac{2\sqrt{2}}{1-(2s^*-1)\mu}\sqrt{s^*}\eta^*
+\frac{2\sqrt{2}-(3\sqrt{2}s^*-1)\mu}{2\sqrt{2}(1-(2s^*-1)\mu)}~\frac{2\|\tilde{x}\|_1}{\sqrt{s^*}}.
\end{align*}
Then (\ref{e3.13}) implies that
\begin{align}\label{e3.17-0}
\|\hat{x}^{DS}&-(x_{\max(s^*)}+\tilde{x})\|_2\nonumber\\
&\leq\frac{2\sqrt{2}}{1-(2s^*-1)\mu}\sqrt{s^*}\eta^*
+\frac{2\sqrt{2}-(3\sqrt{2}s^*-1)\mu}{2\sqrt{2}(1-(2s^*-1)\mu)}\frac{68\sqrt{s^*}\|Ax_{-\max(s^*)}\|_2}{\sqrt{m}\sqrt{s^*}}\nonumber\\
&\leq\frac{2}{1-(2s^*-1)\mu}\sqrt{s^*}\eta^*
+\frac{2\sqrt{2}-(3\sqrt{2}s^*-1)\mu}{2\sqrt{2}(1-(2s^*-1)\mu)}\frac{68\sqrt{{1+(s^*-1)\mu}}\|x_{-\max(s^*)}\|_2}{\sqrt{m}}\nonumber\\
&=\frac{2}{1-(2s^*-1)\mu}\sqrt{s^*}\eta^*+\frac{34\big(2\sqrt{2}-(3\sqrt{2}s^*-1)\mu\big)\sqrt{1+(s^*-1)\mu}}{\sqrt{2}(1-(2s^*-1)\mu)\sqrt{m}}\|x_{-\max(s^*)}\|_2,
\end{align}
where the second inequality follows from (\ref{e3.14}).

Last, combination (\ref{e3.15}) and (\ref{e3.17-0}) yields that
\begin{align*}
\|\hat{x}^{DS}&-x\|_2\\
&\leq\Big(\frac{2}{1-(2s^*-1)\mu}\sqrt{s^*}\eta^*
+\frac{34\big(2\sqrt{2}-(3\sqrt{2}s^*-1)\mu\big)\sqrt{1+(s^*-1)\mu}}{\sqrt{2}(1-(2s^*-1)\mu)\sqrt{m}}\|x_{-\max(s^*)}\|_2\Big)\\
&\hspace*{12pt}+\Big(\frac{34}{\sqrt{m}}+\frac{34\sqrt{1+(s^*-1)\mu}}{\sqrt{m}(\sqrt{1-(s^*-1)\mu})}+\frac{1}{\sqrt{1-(s^*-1)\mu}}\Big)
\sqrt{1+(s^*-1)\mu}\|x_{-max(s^*)}\|_2\\
&\hspace*{12pt}+\|x_{-\max(s^*)}\|_2\\
&=\bigg(\Big(\frac{34}{\sqrt{m}}+\frac{34\sqrt{1+(s^*-1)\mu}}{\sqrt{m}(\sqrt{1-(s^*-1)\mu})}
+\frac{34\big(2\sqrt{2}-(3\sqrt{2}s^*-1)\mu\big)}{\sqrt{m}\sqrt{2}(1-(2s^*-1)\mu)}+\frac{1}{\sqrt{1-(s^*-1)\mu}}\Big)\\
&\hspace*{12pt}\times\sqrt{1+(s^*-1)\mu}+1\bigg)
\|x_{-\max(s^*)}\|_2+\frac{2}{1-(2s^*-1)\mu}\sqrt{s^*}\eta^*\\
&\leq\bigg(\bigg(\frac{34}{\sqrt{m}}\Big(1+\frac{\sqrt{1+(s^*-1)\mu}}{\sqrt{1-(s^*-1)\mu}}
+\frac{2\sqrt{2}-(3\sqrt{2}s^*-1)\mu}{\sqrt{2}(1-(2s^*-1)\mu)}\Big)+\frac{2}{1-(2s^*-1)\mu}\bigg)\\
&\hspace*{12pt}\times\sqrt{1+(s^*-1)\mu}\big(\sqrt{s^*}\eta^*+\|x_{-\max(s^*)}\|_2\big)\\
&\leq\bigg(34\Big(1+\frac{1+(s^*-1)\mu}{1-(2s^*-1)\mu}
+\frac{2-(3s^*-1)\mu}{1-(2s^*-1)\mu}\Big)+\frac{2}{1-(2s^*-1)\mu}\bigg)\\
&\hspace*{12pt}\times\sqrt{1+(s^*-1)\mu}\big(\sqrt{s^*}\eta^*+\|x_{-\max(s^*)}\|_2\big)\\
&=\frac{\big(138-34(4s^*-1)\mu\big)\sqrt{1+(s^*-1)\mu}}{1-(2s^*-1)\mu}\big(\sqrt{s^*}\eta^*+\|x_{-\max(s^*)}\|_2\big).
\end{align*}
\end{proof}

Now, we have made preparations for proving Theorem \ref{OracleGeneral}.

\begin{proof}[Proof of Theorem \ref{OracleGeneral}]
Without loss of generalization, we assume that $\text{supp}(x_{\max(s^*)})\subset S^*$ with $|S^*|=s^*$.
Set $\lambda=\sigma\sqrt{2\log n}$. By Lemma \ref{Probabilityinequality}, event $E=\{z\in\mathbb{R}^n: \|A^*z\|_{\infty}\leq\lambda\}$ occurs with probability at least
$$
1-\frac{1}{2\sqrt{\pi\log n}}.
$$
In the following, we shall assume that event $E$ occurs.

There are three cases to consider, depending on the number of of $x$ standing above the noise level.

\textbf{Case 1: High Noise Level}: Suppose $K(x_{S_0},x)\leq\sigma^2\|x_{S^*}\|_0.$

Then $K(x_{S_0},x)\leq K(x,x)$ and $\|x_{S_0}\|_0\leq\|x_{S^*}\|_0$.
And $\bar{x}=\arg\min_{\xi}K(\xi,x)$ also implies that
$\|\bar{x}\|_0\leq\|x_{S^*}\|_0$. Hence, Lemma \ref{Highnoiselevel}
gives that $\hat{x}^{L}$ satisfies
\begin{align}\label{e3.17}
\|\hat{x}^{L}-x\|_2^2\leq\frac{32(2+\sqrt{2\log n})^2}{\big(\mu(1-4s^{*}\mu)\big)^2}\sum_{j}\min\{\sigma^2,|x(j)|^2\}
\end{align}
with probability at least
$$
1-\frac{1}{2\sqrt{\pi\log n}}.
$$
And for $\hat{x}^{DS}$, Lemma \ref{Highnoiselevel} also implies that
\begin{align}\label{e3.17}
\|\hat{x}^{DS}-x\|_2^2\leq\frac{16(2+\sqrt{2\log n})^2}{\big(1-(2s^{*}-1)\mu\big)^2}\sum_{j}\min\{\sigma^2,|x(j)|^2\}
\end{align}
with probability at least
$$
1-\frac{1}{2\sqrt{\pi\log n}}.
$$

\textbf{Case 2: Low Noise Level}: Suppose $K(x_{S_0},x)>\sigma^2\|x_{S^*}\|_0$ and $\|x_{S_0}\|_0\geq\|x_{S^*}\|_0$.

From Lemma \ref{LQforGaussian}, for the Gaussian measurement ensemble, the requirements of Proposition \ref{Lownoiselevel} are met
with probability at least
$$
1-e^{-m/100}.
$$
Hence, $\hat{x}^L$ satisfies
\begin{align*}
\|\hat{x}^{L}-x\|_2^2&\leq2\Bigg(\frac{\Big(2+34\mu\big(4-3(s^*-1)\mu\big)\Big)\sqrt{1+(s^*-1)\mu}}{\mu(1-4s^*\mu)}\Bigg)^2
\big(s^*(\lambda^*)^2+\|x_{-\max(s^*)}\|_2^2\big)\\
&\leq8\Bigg(\frac{\Big(2+34\mu\big(4-3(s^*-1)\mu\big)\Big)\sqrt{1+(s^*-1)\mu}}{\mu(1-4s^*\mu)}\Bigg)^2\\
&\hspace*{12pt}\times\Big(\frac{5}{4}+\sqrt{2\log n}\Big)^2\big(s^*\sigma^2+\|x_{-\max(s^*)}\|_2^2\big)
\end{align*}
with probability at least
$$
1-e^{-m/100}-\frac{1}{2\sqrt{\pi\log n}}.
$$
The assumption $\|x_{S_0}\|_0\geq\|x_{S^*}\|_0$ implies $S_0\supset S^*$, otherwise $S_0\subset S^*$. Therefore
\begin{align}\label{e3.18}
\sum_{j\in S^*}\min\{\sigma^2,|x(j)|^2\}&=\sigma^2\|x_{S_0\cap S^*}\|_0+\|x_{S^*\backslash S_0}\|_2^2\nonumber\\
&=
\begin{cases}
\sigma^2s^*,&\|x_{S_0}\|_0\geq\|x_{S^*}\|_0\\
\sigma^2\|x_{S_0}\|_0+\|x_{S^*\backslash S_0}\|_2^2, &\|x_{S_0}\|_0<\|x_{S^*}\|_0.
\end{cases}
\end{align}
Hence, $\hat{x}^L$ satisfies
\begin{align}\label{e3.19}
\|\hat{x}^{L}-x\|_2^2
&\leq8\Bigg(\frac{\Big(2+34\mu\big(4-3(s^*-1)\mu\big)\Big)\sqrt{1+(s^*-1)\mu}}{\mu(1-4s^*\mu)}\Bigg)^2
\Big(\frac{5}{4}+\sqrt{2\log n}\Big)^2\nonumber\\
&\hspace*{12pt}\times\bigg(\sum_{j\in S^*}\min\{\sigma^2,|x(j)|^2\}+\|x_{-\max(s^*)}\|_2^2\bigg)
\end{align}
with probability at least
$$
1-e^{-m/100}-\frac{1}{2\sqrt{\pi\log n}}.
$$
And it follows from Proposition \ref{Lownoiselevel} that for $\hat{x}^{DS}$,
\begin{align}\label{e3.20}
&\|\hat{x}^{DS}-x\|_2^2\nonumber\\
&\leq2\bigg(\frac{\big(138-34(4s^*-1)\mu\big)\sqrt{1+(s^*-1)\mu}}{1-(2s^*-1)\mu}\bigg)^2\big(s(\eta^*)^2+\|x_{-\max(s)}\|_2^2\big)\nonumber\\
&\leq2\bigg(\frac{\big(138-34(4s^*-1)\mu\big)\sqrt{1+(s^*-1)\mu}}{1-(2s^*-1)\mu}\bigg)^2\Big(\frac{3}{2}+\sqrt{2\log n}\Big)^2
\big(s\sigma^2+\|x_{-\max(s)}\|_2^2\big)\nonumber\\
&=2\bigg(\frac{\big(138-34(4s^*-1)\mu\big)\sqrt{1+(s^*-1)\mu}}{1-(2s^*-1)\mu}\bigg)^2\Big(\frac{3}{2}+\sqrt{2\log n}\Big)^2\nonumber\\
&\hspace*{12pt}\times\bigg(\sum_{j\in S^*}\min\{\sigma^2,|x(j)|^2\}+\|x_{-\max(s^*)}\|_2^2\bigg)
\end{align}
with probability at least
$$
1-e^{-m/100}-\frac{1}{2\sqrt{\pi\log n}}.
$$

\textbf{Case 3: Medium Noise Level}: Suppose $K(x_{S_0},x)>\sigma^2\|x_{S^*}\|_0$ and $\|x_{S_0}\|_0<\|x_{S^*}\|_0$.

As in \textbf{Case 2}, we have
\begin{align*}
\|\hat{x}^{L}-x\|_2^2
&\leq8\Bigg(\frac{\Big(2+34\mu\big(4-3(s^*-1)\mu\big)\Big)\sqrt{1+(s^*-1)\mu}}{\mu(1-4s^*\mu)}\Bigg)^2\\
&\hspace*{12pt}\times\Big(\frac{5}{4}+\sqrt{2\log n}\Big)^2\big(s^*\sigma^2+\|x_{-\max(s^*)}\|_2^2\big)
\end{align*}
with probability at least
$$
1-e^{-m/100}-\frac{1}{2\sqrt{\pi\log n}}.
$$
Note that
\begin{align*}
s^*\sigma^2+\|x_{-\max(s^*)}\|_2^2&\leq K(x_{S_0},x)+\|x_{-\max(s)}\|_2^2\\
&\leq\sigma^2\|x_{S_0}\|_{0}+2\|x_{S^*\backslash S_0}\|_2^2+3\|x_{-\max(s^*)}\|_2^2\\
&\leq2\sum_{j\in S^*}\min\{\sigma^2,|x(j)|^2\}+3\|x_{-\max(s^*)}\|_2^2\\
&\leq3\bigg(\sum_{j\in S^*}\min\{\sigma^2,|x(j)|^2\}+\|x_{-\max(s^*)}\|_2^2\bigg),
\end{align*}
where the third inequality follows from (\ref{e3.18}). Hence we get that $\hat{x}^{L}$ satisfies
\begin{align}\label{e3.21}
\|\hat{x}^{L}-x\|_2^2
&\leq24\Bigg(\frac{\Big(2+34\mu\big(4-3(s^*-1)\mu\big)\Big)\sqrt{1+(s^*-1)\mu}}{\mu(1-4s^*\mu)}\Bigg)^2\Big(\frac{5}{4}+\sqrt{2\log n}\Big)^2\nonumber\\
&\hspace*{12pt}\times\bigg(\sum_{j\in S^*}\min\{\sigma^2,|x(j)|^2\}+\|x_{-\max(s^*)}\|_2^2\bigg)
\end{align}
with probability at least
$$
1-e^{-m/100}-\frac{1}{2\sqrt{\pi\log n}}.
$$

And for $\hat{x}^{DS}$, we have
\begin{align}\label{e3.22}
\|\hat{x}^{DS}-x\|_2^2&
\leq2\bigg(\frac{\big(138-34(4s^*-1)\mu\big)\sqrt{1+(s^*-1)\mu}}{1-(2s^*-1)\mu}\bigg)^2\Big(\frac{3}{2}+
\sqrt{2\log n}\Big)^2\nonumber\\
&\hspace*{12pt}\times\big(s\sigma^2+\|x_{-\max(s)}\|_2^2\big)\nonumber\\
&\leq6\bigg(\frac{\big(138-34(4s^*-1)\mu\big)\sqrt{1+(s^*-1)\mu}}{1-(2s^*-1)\mu}\bigg)^2\Big(\frac{3}{2}+\sqrt{2\log n}\Big)^2\nonumber\\
&\hspace*{12pt}\times\bigg(\sum_{j\in S^*}\min\{\sigma^2,|x(j)|^2\}+\|x_{-\max(s^*)}\|_2^2\bigg)
\end{align}
with probability at least
$$
1-e^{-m/100}-\frac{1}{2\sqrt{\pi\log n}}.
$$
\end{proof}

%%%%%%%%%%%%%%%%%%%%%%%%%%%%%%%%%%%%%%%%%%%%%%%%%%%%%%%%%%%%%%%%
%%%%%%%%%%%%%%%%%%%%%%%  Section 4 %%%%%%%%%%%%%%%%%%%%%%%%%%%%%%
%%%%%%%%%%%%%%%%%%%%%%%%%%%%%%%%%%%%%%%%%%%%%%%%%%%%%%%%%%%%%%%%
\section{Relationship to Robust Null Space Property\label{s4}}
\hskip\parindent

Besides mutual inherence property, the signal recovery problem has also been well studied in the framework of the (robust) null space property, see \cite{DE2003,CDD2009,S2011,FR2013,F2014,F2017}. In this section, we will study the relationship between mutual incoherence property and the robust null space property (RNSP). The robust null space property with $\ell_2$ bound $\|Ax\|_2$ was first introduced by Sun in \cite{S2011}, which is called sparse approximation property. But this name was first used by Foucart and Rauhut in \cite{FR2013}. And they also introduced the robust null space property with Dantzig selector bound $\|A^*Ax\|_{\infty}$.

\begin{definition}\cite{S2011,FR2013}\label{SparseRieszProperty}
Give $q\geq 1$. The matrix $A\in\mathbb{R}^{m\times n}$ is said to satisfy the $l_q$-robust null space property of order $s$ with $\ell_2$ bound with constant pair $(\rho,\tau)$ , if
\begin{align}\label{lpSparseRieszProperty}
\|x_{\max(s)}\|_q\leq\rho s^{1/q-1}\|x_{-\max(s)}\|_1+\tau\|Ax\|_2 ,
\end{align}
holds for all $x\in\mathbb{R}^n$.

And an $m\times n$ matrix $A$ is said to satisfy the $l_q$-robust null space property of order $s$ with Dantzig selector bound with constant pair $(\rho,\tau)$, if
\begin{align}\label{DSSparseRieszProperty}
\|x_{max(s)}\|_q \leq \rho s^{1/q-1}\|x_{-\max(s)}\|_1+\tau\|A^*Ax\|_\infty
\end{align}
holds for all $x\in\mathbb{R}^n$.
\end{definition}

\begin{remark}
When $x\in\text{Ker}{A}\backslash\{0\}$ and $q=1$, then  (\ref{lpSparseRieszProperty}) and (\ref{DSSparseRieszProperty}) become
\begin{align}\label{NSP}
\|x_{\max(s)}\|_1 \leq \rho \|x_{-\max(s)}\|_1,
\end{align}
which is the null space property introduced in \cite{DE2003,CDD2009}.
\end{remark}

We now point out that the RNSP can be deduced from the MIP.
\begin{theorem}\label{MIP-RNSP}
For any $\iota>1$, suppose matrix $A\in\mathbb{R}^{m\times n}$ satisfies mutual incoherence property with
\begin{align*}
\mu<\frac{\sqrt{\iota-1}}{\sqrt{\iota}(\iota s-1)},
\end{align*}
then
\begin{align}\label{l2RNSP}
\|x_{\max(s)}\|_2&\leq\frac{(\iota s-1)\mu}{\sqrt{(\iota-1)\Big(1-\big((\iota s-1)\mu\big)^2\Big)}}\frac{\|x_{-\max(s)}\|_1}{\sqrt{s}}
+\frac{2\sqrt{1+(\iota s-1)\mu}}{1-\big((\iota s-1)\mu\big)^2}\|Ax\|_2\nonumber\\
&=:\rho\frac{\|x_{-\max(s)}\|_1}{\sqrt{s}}+\tau_1\|Ax\|_2
\end{align}
and
\begin{align}\label{DSRNSP}
\|x_{\max(s)}\|_2&\leq\frac{(\iota s-1)\mu}{\sqrt{(\iota-1)\Big(1-\big((\iota s-1)\mu\big)^2\Big)}}\frac{\|x_{-\max(s)}\|_1}{\sqrt{s}}
+\frac{2\sqrt{\iota s}}{1-\big((\iota s-1)\mu\big)^2}\|A^*Ax\|_{\infty}\nonumber\\
&=:\rho\frac{\|x_{-\max(s)}\|_1}{\sqrt{s}}+\tau_2\|A^*Ax\|_{\infty},
\end{align}
i.e., $A$ satisfies the $\ell_2$-robust null space property of order $s$ with $\ell_2$ bound with constant pair $(\rho,\tau_1)$, and the $\ell_2$-robust null space property of order $s$ with Dantzig selector bound with constant pair $(\rho,\tau_2)$.
\end{theorem}

Before proving Theorem \ref{MIP-RNSP}, we first state the vital lemma-``sparse representation of a polytope", which comes from \cite{CZ2014}.

\begin{lemma}\label{SparseRepresentationofPolytope}
For a positive number $\kappa$ and a positive integer $s$, define the polytope $T(\kappa,s)\subset\mathbb{R}^n$ by
$$
T(\kappa,s)=\{x\in\mathbb{R}^n:\|x\|_{\infty}\leq\kappa,\|x\|_1\leq s\kappa\}.
$$
For any $x\in\mathbb{R}^n$, define the set of sparse vectors $U(\kappa,s,x)\subset\mathbb{R}^n$ by
$$
U(\kappa,s,x)=\{y\in\mathbb{R}^n:\text{supp}(y)\subset\text{supp}(x),\|y\|_0\leq s,\|y\|_1=\|x\|_1,\|y\|_{\infty}\leq\kappa\}.
$$
Then any $x\in T(\kappa,s)$ if and only if $x$ is in the convex hull of $U(\kappa,s,x)$. In particular, any $x\in T(\kappa,s)$ can be expressed as
$$
x=\sum_{i=1}^N\rho_iu^{i}, \text{and}~0\leq\rho_i\leq1,\sum_{i=1}^{N}\rho_i=1,\text{and}~u^{i}\in U(\kappa,s,x).
$$
\end{lemma}

\begin{proof}[Proof of Theorem \ref{MIP-RNSP}]
Our proof is inspired by \cite[Theorem 5]{F2017}.  Suppose $\text{supp}(x_{\max(s)})\subset S$.
Let
$$
\|x_{-\max(s)}\|_1=\|x_{S^c}\|_1=\kappa s.
$$
We partition the set $S^c$ as
$$
S^c=S_1\cup S_2,
$$
where
$$
S_1:=\Big\{j\in S^c: |x_j|>\frac{\kappa}{\iota-1}\Big\},~
S_2:=\Big\{j\in S^c: |x_j|\leq \frac{\kappa}{\iota-1}\Big\}.
$$
And then
$$
x_{-\max(s)}=x_{S_1}+x_{S_2}.
$$
Denote that $s_1=\|x_{S_1}\|_0$. We can derive that $s_1<(\iota-1)s$ from
$$
\|x_{S_1}\|_1> s_1\frac{\|x_{S^c}\|_1}{(\iota-1)s}\geq s_1\frac{\|x_{S_1}\|_1}{(\iota-1)s}.
$$
By
$$
\|x_{S_2}\|_{\infty}\leq \frac{\kappa}{\iota-1}
$$
and
$$
\|x_{S_2}\|_1=\|x_{S^c}\|_1-\|x_{S_1}\|_1\leq \|x_{S^c}\|_1-s_1\frac{\|x_{S^c}\|_1}{(\iota-1)s}=\big((\iota-1)s-s_1\big)\frac{\kappa}{\iota-1},
$$
we get that $x_{S_2}\in T(\kappa, (\iota-1)s-s_1)$. By Lemma
\ref{SparseRepresentationofPolytope}, $x_{S_2}$ can be represented
as the convex hull of $(\iota-1)s-s_1)$-sparse vectors:
$$
x_{S_2}=\sum_{j=1}^N\rho_ju^j,
$$
where $u^j$ is $(\iota-1)s-s_1)$-sparse and
\begin{align*}
\sum_{j=1}^{N}&\rho_j=1,~0\leq\rho_j\leq1,j=1,\ldots,N,\\
\text{supp}&(u^j)\subset \text{supp}(x_{S_2}),\\
\|u^j\|_1&=\|x_{S_2}\|_1, \|u^j\|_{\infty}\leq \kappa.
\end{align*}
Hence,
\begin{align}\label{e4.1}
\|u^j\|_2\leq\|u^j\|_{\infty}\sqrt{\|u^j\|_0}\leq \frac{\kappa}{\iota-1}\sqrt{(\iota-1)s-s_1}\leq\sqrt{\frac{s}{\iota-1}}\kappa=\frac{\|x_{-\max(s)}\|_1}{\sqrt{(\iota-1)s}}.
\end{align}

We observe that
\begin{align}\label{e4.2}
&\langle A(x_{\max(s)}+x_{S_1}),Ax\rangle\nonumber\\
&=\frac{1}{4\delta\mu}\sum_{j=1}^N\rho_j\bigg(\Big\|A\Big(\big(1+\delta\mu\big)(x_{\max(s)}+x_{S_1})+\delta\mu u^j\Big)\Big\|_2^2\nonumber\\
&\hspace*{12pt}-\Big\|A\Big(\big(1-\delta\mu\big)(x_{\max(s)}+x_{S_1})-\delta\mu u^j\Big)\Big\|_2^2\bigg),
\end{align}
By
\begin{align*}
\|(1+&\delta\mu)(x_{\max(s)}+x_{S_1})+\delta\mu u^j\|_0\\
&=\|(1-\delta\mu)(x_{\max(s)}+x_{S_1})-\delta\mu u^j\|_0\\
&\leq s+s_1+\big((\iota-1)s-s_1\big)=\iota s,
\end{align*}
we take $\delta=(\iota s-1)$.

First,  we give an upper bound estimate for the left-hand side of (\ref{e4.2}). It follows from Lemma \ref{MIPLemma} that
\begin{align}\label{e4.3}
\langle A(x_{\max(s)}+x_{S_1}),Ax\rangle&\leq\|A(x_{\max(s)}+x_{S_1})\|_2\|Ax\|_2\nonumber\\
&\leq\sqrt{1+(s+s_1-1)\mu}\|Ax\|_2\|x_{\max(s)}+x_{S_1}\|_2\nonumber\\
&\leq\sqrt{1+(\iota s-1)\mu}\|Ax\|_2\|x_{\max(s)}+x_{S_1}\|_2.
\end{align}
And it follows from $\|v\|_p\leq(\|v\|_0)^{1/p-1/q}\|v\|_q$ for $0<p\leq q\leq\infty$ that
\begin{align}\label{e4.4}
\langle A(x_{\max(s)}+x_{S_1}),Ax\rangle&\leq\|x_{\max(s)}+x_{S_1}\|_1\|A^*Ax\|_{\infty}\nonumber\\
&\leq\sqrt{s+s_1}\|A^*Ax\|_\infty\|x_{\max(s)}+x_{S_1}\|_2\nonumber\\
&\leq\sqrt{\iota s}\|A^*Ax\|_\infty\|x_{\max(s)}+x_{S_1}\|_2.
\end{align}

On the other hand, by mutual incoherence property, the right-hand side of (\ref{e4.2}) is bounded from below by
\begin{align}\label{e4.5}
\frac{1}{4\delta\mu}&\sum_{j=1}^N\rho_j\bigg(\big(1-\delta\mu\big)\Big(\big(1+\delta\mu\big)^2\|x_{\max(s)}+x_{S_1}\|_2^2
+\big((2s-1)\mu\big)^2\|u^j\|_2^2\Big)\nonumber\\
&\hspace*{12pt}-(1+\delta\mu)\Big(\big(1-\delta\mu\big)^2\|x_{\max(s)}+x_{S_1}\|_2^2+\big(\delta\mu\big)^2 \|u^j\|_2^2\Big)\bigg)\nonumber\\
&=\frac{1}{4\delta\mu}\sum_{j=1}^N\rho_j\Big(2\delta\mu(1+\delta\mu)(1-\delta\mu)\|x_{\max(s)}+x_{S_1}\|_2^2\nonumber\\
&\hspace*{12pt}-2\big(\delta\mu\big)^3\|u^j\|_2^2\Big)\nonumber\\
&\geq\frac{1-\big(\delta\mu\big)^2}{2}\|x_{\max(s)}+x_{S_1}\|_2^2-\frac{\big(\delta\mu\big)^2}{2}\bigg(\frac{\|x_{-\max(s)}\|_1}{\sqrt{(\iota-1)s}}\bigg)^2,
\end{align}
where the last inequality follows from (\ref{e4.1}).

Let $Y=\|x_{\max(s)}+x_{S_1}\|_2$. Then the combination of the two bounds yields
\begin{align*}
\frac{1-\big(\delta \mu\big)^2}{2}Y^2
-\sqrt{1+\delta\mu}\|Ax\|_2Y-\frac{\big(\delta\mu\big)^2}{2}\bigg(\frac{\|x_{-\max(s)}\|_1}{\sqrt{(\iota-1)s}}\bigg)^2\leq0
\end{align*}
and
\begin{align*}
\frac{1-\big(\delta\mu\big)^2}{2}Y^2
-\sqrt{\iota s}\|A^*Ax\|_\infty Y-\frac{\big(\delta\mu\big)^2}{2}\bigg(\frac{\|x_{-\max(s)}\|_1}{\sqrt{(\iota-1)s}}\bigg)^2\leq0.
\end{align*}
So far, we have obtained two  second-order inequalities for $Y$.
Owing to $(1-\big(\delta\mu\big)^2)/2>0$, we can solve above two inequality and get
\begin{align*}
\|x_{\max(s)}\|_2\leq Y
&\leq\frac{2\sqrt{1+\delta\mu}}{1-\big(\delta\mu\big)^2}\|Ax\|_2+\frac{\delta\mu}{\sqrt{(\iota-1)\big(1-(\delta\mu)^2\big)}}\frac{\|x_{-\max(s)}\|_1}{\sqrt{s}}\\
&=:\tau_1\|Ax\|_2+\rho\frac{\|x_{-\max(s)}\|_1}{\sqrt{s}}
\end{align*}
and
\begin{align*}
\|x_{\max(s)}\|_2\leq Y
&\leq\frac{2\sqrt{\iota s}}{1-\big(\delta\mu\big)^2}\|A^*Ax\|_{\infty}
+\frac{\delta\mu}{\sqrt{(\iota-1)\big(1-(\delta\mu)^2\big)}}\frac{\|x_{-\max(s)}\|_1}{\sqrt{s}}\\
&=:\tau_2\|A^*Ax\|_\infty+\rho\frac{\|x_{-\max(s)}\|_1}{\sqrt{s}},
\end{align*}
which are the desired inequalities with
$$
\rho=\frac{\delta\mu}{\sqrt{(\iota-1)\big(1-(\delta\mu)^2\big)}}=\frac{(\iota s-1)\mu}{\sqrt{(\iota-1)\Big(1-\big((\iota s-1)\mu\big)^2\Big)}}<1
$$
when
$$
\mu<\frac{\sqrt{\iota-1}}{\sqrt{\iota}(\iota s-1)},
$$
and with
$$
\tau_1=\frac{2\sqrt{1+(\iota s-1)\mu}}{1-\big((\iota s-1)\mu\big)^2},~
\tau_2=\frac{2\sqrt{\iota s}}{1-\big((\iota s-1)\mu\big)^2}.
$$
\end{proof}

\begin{remark}\label{MIP-RNSPremark}
If we take $\iota=3/2$, then we get that  then $A$ satisfies the $\ell_2$-robust null space property of order $s$ can be deduced from the coherence condition $\mu<1/\big(\sqrt{3}(3s/2-1)\big)$.
\end{remark}

%%%%%%%%%%%%%%%%%%%%%%%%%%%%%%%%%%%%%%%%%%%%%%%%%%%%%%%%%%%%%%%
%%%%%%%%%%%%%%%%%%%%%% Section 5 %%%%%%%%%%%%%%%%%%%%%%%%%%%%%%
%%%%%%%%%%%%%%%%%%%%%%%%%%%%%%%%%%%%%%%%%%%%%%%%%%%%%%%%%%%%%%%

\section{Conclusions and Discussion \label{s5}}
\hskip\parindent

In this paper, we first obtain an upper bound of error of the original signal $x$ and reconstructed signals $\hat{x}^{L}$ of Lasso model under the mutual incoherence property condition $\mu<1/(4s)$ (Theorem \ref{LassoTheorem}). And we also obtain the lower bound estimate of $\|\hat{x}^{L}-x\|_2$ and $\|\hat{x}^{DS}-x\|_2$ for sparse signal $x$ for Gaussian noise observations in the sense of expectation and probability (Theorem \ref{ExpectionLowerbound} and Theorem \ref{ProbabilityLowerbound}). For Lasso model, we also get the oracle inequalities for both sparse signal and non-sparse signal under the condition $\mu<1/(4s)$ (Theorem \ref{OracleSparseLasso} and Theorem \ref{OracleGeneral}). And as a supplement of \cite[Theorem 4.1]{CWX2010}, we also give an oracle inequality of Dantzig selector for non-sparse signal under the condition $\mu<1/(2s-1)$ (Theorem \ref{OracleGeneral}). In the last section, we investigate the relationship between mutual incoherence property and robust null space property, we find that the $\ell_2$ robust null space property of order $s$ can be deduced from the condition $\mu<\sqrt{\iota-1}/\big(\sqrt{\iota}(\iota s-1)\big)$ for any fixed $\iota>1$ (Theorem \ref{MIP-RNSP}). Therefore, these results in our paper may guide the practitioners to study Lasso and  oracle inequalities in framework of MIP.

However, Cai, Wang and Xu \cite{CWX2010} showed that the MIP condition $\mu<1/(2s-1)$ is sharp for stable recovery of $s$-sparse signals in the presence of noise. Therefore, our condition $\mu<1/(4s)$ for Lasso model (Theorem \ref{LassoTheorem}) and $\mu<1/\big(\sqrt{3}(3s/2-1)\big)$ for robust null space property (Theorem \ref{MIP-RNSP} and Remark \ref{MIP-RNSPremark}) may be not sharp. Obtaining the sharp bound of MIP is one direction of our future research.

%%%%%%%%%%%%%%%%%%%%%%%%%%%%%%%%%%%%%%%%%%%%%%%%%%%%%%%%%%%%%%%
%%%%%%%%%%%%%%%%%%%%%%%%%%%%%%%%%%%%%%%%%%%%%%%%%%%%%%%%%%%%%%
%%%%%%%%%%%%%%%%%%%%%%%%%%%%%%%%%%%%%%%%%%%%%%%%%%%%%%%%%%%%%%

\textbf{Acknowledgement}: Wengu Chen is supported by National Natural Science Foundation of China (No. 11371183).

%%%%%%%%%%%%%%%%%%%%%%%%%%%%%%%%%%%%%%%%%%%%%%%%%%%%%%%%%%%%%%%%%%%%%%%%%%%%%%%%%%%%%%%%%%%

%%%%%%%%%%%%%%%%%%%%%%%%%%%%%%%%%%%   Bibliography  %%%%%%%%%%%%%%%%%%%%%%%%%%%%%%%%%%%%%%%%

%%%%%%%%%%%%%%%%%%%%%%%%%%%%%%%%%%%%%%%%%%%%%%%%%%%%%%%%%%%%%%%%%%%%%%%%%%%%%%%%%%%%%%%%%%%%


\begin{thebibliography}{SIBL}
\vspace{-0.3cm}

\bibitem{BRT2009}Bickel P J., Ritov Y., Tsybakov A B. Simultaneous analysis of Lasso and Dantzig selector. Ann. Statist., 2009, 37: 1705-1732.

\bibitem{CXZ2009}Cai T T., Xu G., Zhang J. On recovery of sparse signals via $\ell_1$ minimization.
IEEE Trans. Inform. Theory, 2009, 55: 3388-3397.

\bibitem{CWX2010}Cai T T., Wang L., Xu G. Stable recovery of sparse signals and an oracle inequality. IEEE Trans. Inform. Theory, 2010, 56: 3516-3522.

\bibitem{CZ2013}Cai T T., Zhang A. Compressed sensing and affine rank minimization under restricted isometry. IEEE Trans. Inform. Theory, 2013, 61: 3279-3290.

\bibitem{CZ2014}Cai T T., Zhang A. Sparse representation of a polytope and recovery of sparse signals and low-rank matrices. IEEE Trans. Inform. Theory, 2014, 60: 122-132.

\bibitem{CP2011}Cand\`{e}s E J., Plan Y. Tight oracle inequalities for low-rank matrix recovery from a minimal number of noisy random measurements. IEEE Trans. Inform. Theory, 2011, 57: 2342-2359.

\bibitem{CRT2006}Cand\`{e}s E J., Romberg J K., Tao T. Stable signal recovery from incomplete and inaccurate measurements. Comm. Pure Appl. Math., 2006, 59: 1207-1223.

\bibitem{CT2005}Cand\`{e}s, E.J., Tao, T.: Decoding by linear programming. IEEE Trans. Inform. Theory \textbf{51}, 4203-4215(2005).

\bibitem{CT2006}Cand\`{e}s E J., Tao T. Near optimal signal recovery from random projections: universal
encoding strategies? IEEE Trans. Inform. Theory, 2006, 52: 5406-5425.

\bibitem{CT2007}Cand\`{e}s E J., Tao T. The dantzig selector: Statistical estimation when
$p$ is much larger than $n$. Ann. Statist., 2007, 33: 2313-2351.

\bibitem{CDS1998}Chen S S., Donoho D L., Saunders M A. Atomic Decomposition by Basis Pursuit, SIAM. J. Sci. Comput., 1998, 20: 33-61.

\bibitem{CDD2009}Cohen A., Dahmen W., DeVore R. Compressed sensing and best $k$-term approximation. J. Amer. Math. Soc., 2009, 22: 211-231.

\bibitem{D2013}De, Castro, Y. A remark on the lasso and the Dantzig selector. Statistics and Probability Letters, 2013, 83: 304-314.

\bibitem{D2006}Donoho D L. Compressed Sensing,. IEEE Trans. Inform. Theory, 2006, 52: 1289-1306.

\bibitem{DE2003}Donoho D L., Elad M. Optimally sparse representations in general (nonorthogonal) dictionaries
via $\ell_1$ minimization. Proc. Natl. Acad. Sci. USA., 2003, 100: 2197-2202.

\bibitem{DET2005}Donoho D L., Elad M. Temlyakov V.N. Stable recovery of sparse overcomplete representations in the presence of noise. IEEE Trans. Inform. Theory, 2005, 52: 6-18.

\bibitem{DH2001}Donoho D L., Huo X. Uncertainty principles and ideal atomic decomposition. IEEE Trans. Inform. Theory, 2001, 47:
2845-2862.

\bibitem{DJ1994}Donoho D L., Johnstone I M. Ideal spatial adaptation by wavelet shrinkage. Biometrika., 1994, 81: 425-455.

\bibitem{EMR2007}Elad M., Milanfar P., Rubinstein R. Analysis versus synthesis in signal priors. Inverse Problems., 2007, 23: 947-968.

\bibitem{ES2005}Erickson S., Sabatti C. Empirical Bayes estimation of a sparse vector of gene expression changes. Statist. Appl. Genetics Mol. Biol.,
2005, 4: 1-27.

\bibitem{F2014}Foucart S. Stability and robustness of $\ell_1$-minimizations with Weibull matrices and redundant dictionaries. Linear Algebra Appl., 2014, 441: 4-21.

\bibitem{F2017}Fourcat S. Flavors of Compressive Sensing. International Conference Approximation Theory, Approximation Theory XV: San Antonio 2016, ed. by Fasshauer, G.E., Schumaker, L.L., 61-104.

\bibitem{FR2013}Foucart S., Rauhut H. A mathematical introduction to compressive sensing, Applied and
Numerical Harmonic Analysis Series, Birkh\"{a}user, Basel(2013).

\bibitem{HS2009}Herman M., Strohmer T. High-resolution radar via compressed sensing. IEEE Trans. Signal Process., 2009, 57: 2275-2284.

\bibitem{LC1998}Lehmann E., Casella G. Theory of Point Estimation. Springer Verlag Press, New York(1998).

\bibitem{LL2014}Lin J H.,  Li S. Sparse recovery with coherent tight frames via analysis Dantzig selector and analysis LASSO. Appl. Comput.Harmon.Anal., 2014, 37: 126-139.

\bibitem{LDP2007}Lustig, M., Donoho, D.L.,  Pauly, J.M.: Sparse MRI: The application of compressed sensing for
rapid MR imaging. Magnetic Resonance in Medicine \textbf{58}, 1182-1195(2007).

\bibitem{PVMH2008}Parvaresh F., Vikalo H., Misra S., Hassibi B. Recovering sparse signals
using sparse measurement matrices in compressed DNA microarrays. IEEE J. Sel. Top. Signal Process., 2008, 2: 275-285.

\bibitem{SV2008}Schnass K., Vandergheynst P. Dictionary Preconditioning for Greedy Algorithms. IEEE Trans. Signal Process., 2008, 56: 1994-2002.

\bibitem{SHB2015}Shen Y., Han B., Braverman E. Stable recovery of analysis based approaches. Appl. Comput. Harmon. Anal., 2015, 29: 161-172.

\bibitem{S2011}Sun, Q. Sparse approximation property and stable recovery of sparse signals from noisy measurements. IEEE Trans. Signal Process., 2011,  59: 5086-5090.

\bibitem{TEBN2014}Tan Z., Eldar Y C., Beck A., Nehorai A. Smoothing and decomposition for analysis sparse recovery. IEEE Trans. Signal Process., 2014, 62: 1762-1774.

\bibitem{THER2010}Taub\"{o}ck G., Hlawatsch F., Eiwen D., Rauhut H. Compressive estimation of doubly selective
channels in multicarrier systems: leakage effects and sparsity-enhancing processing. IEEE J.
Sel. Top. Signal Process., 2010, 4: 255-271.

\bibitem{T1996}Tibshirani R. Regression shrinkage and selection via the lasso. J. Roy. Stat. Soc. B., 1996, 58: 267-288.

\bibitem{T2004}Tropp J A. Greed is good: algorithmic results for sparse approximation. IEEE Trans. Inform. Theory, 2004, 50: 2231-2242.

\bibitem{T2009}Tseng P. Further results on a stable recovery of sparse overcomplete representations in the presence of noise. IEEE Trans. Inform. Theory, 2009,
55: 888-899.

\bibitem{VAHBPL2010}Vasanawala S., Alley M., Hargreaves B., Barth R., Pauly J., Lustig M. Improved pediatric
MR imaging with compressed sensing. Radiology., 2010, 256: 607-616.


\bibitem{W2010}Wojtaszczyk P. Stability and instance optimality for Gaussian measurements in compressed
sensing. Found. Comput. Math., 2010, 10: 1-13.

\bibitem{XL2016}Xia Y., Li S. Analysis recovery with coherent frames and correlated measurements. IEEE Trans. Inform. Theory, 2016, 62: 6493-6507.

\bibitem{ZL2017}Zhang R., Li S. A Proof of Conjecture on Restricted Isometry Property Constants $\delta_{tk}~(0<t<\frac{4}{3})$. IEEE Trans. Inform. Theory (2017), DOI: 10.1109/TIT.2017.2705741.

\bibitem{ZYY2016}Zhang H., Yan M., Yin W. One condition for solution uniqueness and robustness of both $l_1$-synthesis and $l_1$-analysis minimizations. Adv. Comput. Math., 2016, 42: 1381-1399.

\end{thebibliography}
\end{document}